\documentclass[journal]{IEEEtran}

\usepackage{booktabs}
\usepackage{amsmath,amssymb,amsfonts}
\usepackage{algorithmic}
\usepackage{graphicx}
\usepackage{algorithm,algorithmic}
\usepackage{hyperref}
\usepackage{textcomp}
\usepackage{subcaption}
\newtheorem{theorem}{\bf Theorem}[section]
\newtheorem{lemma}{\bf Lemma}[section]
\newtheorem{remark}{\bf Remark}[section]
\newtheorem{assumption}{\bf Assumption}[section]

\newtheorem{corollary}{\bf Corollary}[section]
\newtheorem{problem}{\bf Problem}[section]
\usepackage{tikz}
\usetikzlibrary{arrows.meta}

\newcommand{\bm}[1]{\boldsymbol{#1}}
\newcommand{\mc}[1]{\mathcal{#1}}
\newcommand{\mb}[1]{\mathbb{#1}}

\newcommand{\m}[1]{\mathbf{#1}}
\def\BibTeX{{\rm B\kern-.05em{\sc i\kern-.025em b}\kern-.08em
    T\kern-.1667em\lower.7ex\hbox{E}\kern-.125emX}}
\begin{document}
\title{Adaptive stabilization of a leaderless bearing-constrained formation with disturbances}
\author{Minh Hoang Trinh$^{*}$, 
Chuong Nguyen, 
Quoc Van Tran, Tuynh Van Pham 
\thanks{$^*$Corresponding author.}
\thanks{M. H. Trinh is with East Asia University of Technology, Hanoi, Vietnam (e-mail: \texttt{minhtrinh@ieee.org}). 
C. Nguyen is with Viterbi School of Engineering – Department of Aerospace and Mechanical Engineering, University of Southern California, CA, USA  and SAF Lab., Amazon, USA (e-mail: \texttt{vanchuong@usc.edu}). 
Q. V. Tran is with Department of Mechatronics, School of Mechanical Engineering, Hanoi University of Science and Technology, Hanoi, Vietnam (e-mail: \texttt{quoc.tranvan@hust.edu.vn}). T. V. Pham is with Department of Automation Engineering, School of Electrical and Electronic Engineering, Hanoi University of Science and Technology, Hanoi, Vietnam (e-mail: \texttt{tuynh.phamvan@hust.edu.vn})
}
}

\maketitle

\begin{abstract}
In this paper, we consider the problem of regulating and maintaining a target formation characterized by a set of bidirectional bearing constraints under disturbances. The agents in the formation are modeled by single integrators with {deterministic} bounded continuous disturbances of which {both actual function} and upper bound are unavailable for the control design. Due to the time-varying disturbances, the target formation is time-varying. We propose displacement/bearing-only based adaptive sliding mode control laws to uniformly asymptotically stabilize the {time-varying} target formation and reject the matched disturbances {, and provide sufficient conditions to prevent possible collisions. A time-dependent adaptive design method is provided to convert the displacement-based control law to a bearing-only one.} In addition, to alleviate chattering phenomena from sliding mode control, smooth adaptive control laws are introduced to guarantee uniform global boundedness of the desired formation. Finally, simulation results are given to support the analysis.
\end{abstract}


\section{Introduction}
\label{sec:1}
\IEEEPARstart{O}{ver} the past decade, formation control research has attracted significant interest from the robotics and control systems communities \cite{anderson2008rigid, oh2015survey}. Formations of unmanned robots have been developed for both civilian and military applications, including truck platooning, drone squads for highway monitoring, error checking in solar panel fields, precision agriculture, search and rescue, unmanned underwater vehicles for seabed mapping, and satellite formations for positioning and remote sensing. Moreover, research progresses in formation control are often applicable to its dual problem—the sensor network localization problem \cite{Eren2003sensor,Zhao2016aut}.

Let a group of autonomous agents (AGVs, UAVs, UUVs, satellites,...) {, which  will be referred to as a formation,} be arbitrarily  {allocated} in the space. Formation control focuses on designing control algorithms to move a formation from  {an} initial configuration along a continuous trajectory to  {the set of } target configuration, which is specified by a set of geometric constraints. Each agent in the group has  {local-, onboard-controller and sensing equipment}, and the formation acquisition task is usually performed in an ambient space that is significantly larger in comparison to agents' size. As a result, in a canonical formulation \cite{oh2015survey}, the agents can be modeled as single-integrator dynamics, with the control input being its velocity. The agent's internal controller then determines actual actuator signals to realize the commanded velocity.

An important requirement in formation control is that the control law must be decentralized or distributed. To meet this requirement, each agent is considered as an independent system capable of measuring and/or communicating (via wireless channels) certain geometric variables related to neighboring agents. These variables may include global or relative positions, inter-agent distances, offset angles, and direction vectors. Based on the information accessible to each agent, formation control strategies are categorized into position-, displacement-, distance-, bearing-, and angle-based methods \cite{oh2015survey}. Generally, the complexity of the formation control problem is inversely proportional to the amount of information each agent can access, measure, or exchange.

Recently, formation control algorithms based on bearing vectors have gained increasing attention \cite{Su2026bearing}. The concept of control and localization using bearing measurements is inspired by the vision-based navigation observed in animal behavior \cite{Montijano2016,tron2016distributed}. Bearing vectors can be acquired from onboard cameras, which provide data on the relative direction between agents. Compared to other approaches, bearing-based control reduces the number of sensors required per agent and minimizes the dependence on a global reference frame \cite{Doodeman2025leader}. Furthermore, since cameras are passive sensors, this approach is well-suited for military applications where signal emission must be limited \cite{Ye2017bearing}.

Bearing rigidity theory and bearing-only stabilization of leaderless formations of single-integrators have been studied in \cite{zhao2015bearing,Eric2014,Tran2018TCNS}. While formation acquisition is a fundamental task, it is necessary to consider disturbances and uncertainties in any practical formation \cite{Aldarmini2022}. For formations with leaders, \cite{Zhao2015translational,Nguyen2024TCyber} showed that formation  {stabilization with disturbance rejection} and formation tracking  {are the same problem}. Accordingly, different robust bearing-only control strategies for formation acquisition or tracking were proposed \cite{Zhao2020bearing,trinh2021robust,Bae2019leader, Trinh2021LCSS,Jacob2024}. However,  {existing works either assume a-priori knowledge of the disturbance's upper bound and/or consider formations with several leaders to fix the scale of the target formation.} For example, the authors in \cite{Zhao2020bearing,Li2021bearing,trinh2021robust} considered the formation tracking problem with constant leaders' velocity. Bearing-only formation under  disturbances was considered in \cite{Jacob2024}, where the upper bound of the disturbances  {is available}. Bearing-only formation disturbance rejection and tracking with unknown disturbance's upper bound using adaptive sliding-mode control were proposed in \cite{Nguyen2024TCyber}. The work \cite{Trinh2021LCSS} proposed a finite-time bearing-only formation tracking strategy with unknown leader velocities based on the inverse of the minimum eigenvalue of local orthogonal projection matrices. However, the method can only be used with acyclic leader-follower graph topologies and under assumption that immediate leaders never take up collinearity positions. Bearing-only leaderless formation stabilization was considered in \cite{Bae2022distributed,Bae2020tac}, in which a robust adaptive gradient control strategy was employed to suppress the effects of the unknown disturbance and maintain the bearing error within a certain upper bound. However, the upper bound of the bearing error level in \cite{Bae2022distributed,Bae2020tac} is state-dependent and cannot be  {determined from on initial condition}. The issue of actuator faults and time-delay were considered in \cite{Wu2024distributed}.  {It is also worth mentioning recent works \cite{Jing2019angle,Chen2020angle,Chen2022maneuvering,Cheah2025finite,Garanayak2025}  on stabilizing a desired formation using only local bearing measurements, based on sign-elevated angle rigidity and angle rigidity theory.} The control of leaderless formation with mismatched disturbances, however, has not been considered in these works.

This paper proposes distributed control strategies to achieve a desired bearing-based formation in a leaderless multi-agent system of single-integrator dynamics, under the presence of deterministic, unknown but bounded disturbances. Two sensing scenarios are considered for each agent: (i) relative position measurements with respect to neighboring agents (bearing-based control), and (ii) bearing vector measurements only (bearing-only control). To stabilize the desired formations, we develop adaptive sliding-mode control strategies that guarantee fast convergence and robustness against disturbances \cite{Obeid2018barrier}. The main idea is to construct an adaptive mechanism that increases the control gain in response to the formation error, thereby progressively estimating the unknown upper bound of the disturbance \cite{Oliveira2016adaptive,Roy2020adaptive,Cao2011TAC}. In particular, each bearing constraint specifying the desired formation is associated with an individual adaptive gain updating according to the error between the actual and desired bearing.  {With this adaptation strategy,} the agents can eventually move to the set of target formation using vision sensors. Since the proposed formation control laws cannot reject mismatched disturbances, the target configuration is time-varying and disturbance-dependent. We characterize the time-varying target formation and establish uniform asymptotic stability of the target formation under each proposed control scheme. Further, a new time-dependent adaptive strategy is introduced to convert the bearing-based control law into a bearing-only one. Under the mere assumption that the agents maintain positive separations  {while} the formation evolves, the new bearing-only control law guarantees finite time formation acquisition with bounded control inputs. Finally, smooth adaptive control laws are developed to mitigate the chattering phenomenon and guarantee that the target formation remains uniformly ultimately bounded for both schemes.

The rest of this paper is organized as follows. Section \ref{sec:2} presents theoretical background on bearing rigidity theory and formulates the problems. Sections \ref{sec:3} and \ref{sec:4} focus on  the bearing-based and the bearing-only formation control laws, respectively. Section \ref{sec:5} provides numerical simulations. Lastly, section \ref{sec:6} concludes the paper.

\section{Problem formulation}
\label{sec:2}
\subsection{Bearing rigidity theory}
Consider a set of $n$ non-collocated points in the $d$-dimensional space ($n\ge 2$, $d\ge2$) $\m{p}_i\in \mathbb{R}^d$, with $\m{p}_i \neq \m{p}_j,~\forall i \neq j,~ i, j \le n$. A framework (or a formation) in the $d$-dimensional space ($\mc{G},\m{p}$) is defined by an undirected graph $\mc{G}=(\mc{V},\mc{E})$ and a configuration $\m{p} = [\m{p}_1^\top,\ldots, \m{p}_n^\top]^\top \in \mathbb{R}^{dn}$. 

Specifically, $\mc{V}=\{1,\ldots,n\}$ is the set of $n$ vertices and $\mc{E} \subset \mc{V}\times \mc{V}$ is the set of $m$ edges. Each edge $(i,j)$ in $E$ joins two distinct vertices $i,j\in V$, and we do not distinguish between $(i,j)$ and $(j,i)$ as they denote a same edge. The neighbor set of a vertex $i$, denoted by $\mc{N}_i$, consists of all vertices $j$ that are directly joined to $i$ by edges $(i,j)\in E$.

We will also use $e_1,\ldots,e_m$ to denote the edges of $\mc{E}$ indexing in a certain order. With $e_k = (i,j) \in \mc{E}$, $k=1,\ldots,m$, the bearing vector from $\m{p}_i$ to $\m{p}_j$ is defined by $\m{g}_{ij} = \frac{\m{z}_{ij}}{\Vert \m{z}_{ij}\Vert}$, where $\m{z}_{ij} = \m{p}_j-\m{p}_i$ is the displacement vector between $\m{p}_i$ and $\m{p}_j$. 

Choosing an arbitrarily orientation and indexing of the edges in $\mc{E}$, we can define a corresponding incidence matrix $\m{H}=[h_{ki}] \in \mb{R}^{m \times n}$ of the graph,  {with the entries}
\begin{align}
    h_{ki} = \left\lbrace \begin{array}{rl}
        1, &  e_k = (i,j),\\
       -1, & e_k = (j,i),\\
        0, & \text{otherwise.}
    \end{array} \right.
\end{align}
Using the incidence matrix, we can write the displacement vector as $\m{z}=[\ldots,\m{z}_{ij}^\top,\ldots]^\top= [\m{z}_1^\top,\ldots,\m{z}_m^\top]^\top = \bar{\m{H}}\m{p}$. For each vector $\m{g}_{ij}$, a corresponding orthogonal projection matrix $\m{P}_{\m{g}_{ij}} = \m{I}_d-\m{g}_{ij} \m{g}_{ij}^\top$ can be defined. The matrix $\m{P}_{\m{g}_{ij}}$ is symmetric positive semidefinite, idempotent $\m{P}_{\m{g}_{ij}}^2=\m{P}_{\m{g}_{ij}}$, and has spectrum $\{0,1,\ldots,1\}$. For any vector $\m{y} \in \mb{R}^d$, $\m{P}_{\m{g}_{ij}}\m{y}=\m{0}_d$ if and only if two vectors $\m{y}$ and $\m{g}_{ij}$ are parallel.

Two formations $(\mc{G}, \m{p})$ and $(\mc{G}, \m{p}')$ are bearing equivalent if and only if $\m{P}_ {\m{g}_{ij}} (\m{p}_j'-\m{p}_i') = \m{0}_d, ~\forall (i, j) \in \mc{E}$. The formations $(\mc{G}, \m{p})$ and $(\mc{G}, \m{p}')$ are bearing congruent if and only if $\m{P}_{\m{g}_{ij}} (\m{p}_j'-\m{p}_i') = \m{0}_d,~ \forall i, j \in \mc{V}, i \neq j$. A formation ($\mc{G}, \m{p}$) is globally bearing rigid if any formation having the same bearing constraints with $(\mc{G}, \m{p})$ is also bearing congruent with $(\mc{G},\m{p})$. Let $\m{g} = [\m{g}_1^\top,\ldots,\m{g}_m^\top]^\top \in \mathbb{R}^{dm}$, the bearing rigidity matrix is defined by
\begin{align}
\m{R}_{\rm b}=\m{R}_{\rm b}(\m{p})= \frac{\partial \m{g}}{\partial \m{p}} \triangleq \text{blkdiag}\left(\frac{\m{P}_{\m{g}_k}}{\Vert \m{z}_k\Vert}\right)\Bar{\m{H}}\in \mathbb{R}^{dm\times dn},
\end{align}
where $\bar{\m{H}}= \m{H}\otimes \m{I}_d$. A formation is infinitesimally bearing rigid in $\mathbb{R}^d$ if and only if $\text{rank}(\m{R}_{\rm b}) = dn-d-1$, this means $\text{ker}(\m{R}_{\rm b}) = \text{im}([\m{1}_n\otimes \m{I}_d, \m{p}-\m{1}_n\otimes \bar{\m{p}}])$, where $\bar{\m{p}}=\frac{1}{n}(\m{1}^\top_n\otimes \m{I}_d)\m{p}$ denotes the formation's centroid. The augmented bearing rigidity matrix $\tilde{\m{R}}_{\rm b}=\text{blkdiag}(\Vert \m{z}_k\Vert \otimes \m{I}_d)\m{R}_{\rm b}=\text{blkdiag}(\m{P}_{\m{g}_k})\Bar{\m{H}}$ has the same rank and the same null space as $\m{R}_{\rm b}$ but does not contain information of the relative distances between the agents $\|\m{z}_k\|$. The bearing Laplacian ${\m{L}}_{\rm b} = \tilde{\m{R}}_{\rm b}^\top \tilde{\m{R}}_{\rm b}$ is symmetric positive semi-definite  and is a special type of matrix-weighted Laplacian \cite{Trinh2026Springer}[Chapter 3]. For an infinitesimally rigid framework, $\tilde{\m{L}}_{\rm b}$ has exactly $d+1$ zero eigenvalues, and ker$(\tilde{\m{L}}_{\rm b})=$ker$(\tilde{\m{R}}_{\rm b})$.

\subsection{Problem formulation}
We consider an $n$-agent system in the $d$-dimensional space. Assume that each agent $i$ maintains a local coordinate system ${}^i\Sigma$, which is attached to each agent. The local coordinate systems are aligned, however, their origins may differ from each other.

The agent's motions are modeled by the equation
\begin{align} \label{eq:agent_model}
    \dot{\m{p}}_i = \m{u}_i + \m{d}_i,\, i = 1, \ldots, n,
\end{align}
where $\m{p}_i$ and $\m{d}_i$ correspondingly denote the position of the agent and the disturbance, both written in a global coordinate system. The disturbance vector $\m{d}=[\m{d}_1^\top,\ldots,\m{d}_n^\top]^\top$ is a bounded, time-varying, uniformly continuous function. The upper bound of the disturbance is denoted as $\sup_{t\ge 0}\|\m{d}\|_{\infty} = \beta>0$, and this information is unknown to each agent. 

Let $\m{p}^*=[(\m{p}_1^*)^\top,\ldots,(\m{p}^*_n)^\top]^\top \in \mb{R}^{dn}$ denote a desired formation (or target formation), which does not have any collocated points. The desired formation induces a set of desired bearing vectors $\Gamma =\{\m{g}_{ij}^* = \frac{\m{p}_j^*-\m{p}_i^*}{\|\m{p}_j^*-\m{p}_i^*\|}\}_{(i,j)\in \mc{E}}$. It is assumed that $\m{p}^*$ is infinitesimally bearing rigid in $\mb{R}^{d}$. We define the set of formations which are bearing congruent to $\m{p}^*$ as $\mc{D} = \{\m{p}\in\mb{R}^{dn}|~\tilde{\m{R}}_{\rm b}^*\m{p}=\m{0}_{dm},\m{p}_i\neq \m{p}_j\forall i,j\in \mc{V},i\ne j\}$, where $\tilde{\m{R}}_{\rm b}^* \triangleq \tilde{\m{R}}_{\rm b}(\m{p}^*)=\m{P}_{\m{g}^*}\bar{\m{H}}$ and $\m{P}_{\m{g}^*} \triangleq \text{blkdiag}(\m{P}_{\m{g}_k^*})$.

To control the local bearing constraints to match with those defined from the target formation, the agents need to sense some geometric variables with regard to their neighbors. Two types of relative sensing variables, namely, the displacements $\m{z}_{ij} = \m{p}_j - \m{p}_i,$ and the bearing vectors $\m{g}_{ij} = \frac{\m{p}_j - \m{p}_i}{\|\m{p}_j - \m{p}_i\|},\forall j \in \mc{N}_i$ will be considered in this paper.

\begin{problem} \label{problem:1}
Suppose that each agent can sense the displacements with regard to its neighbors. Design control law for each agent using the available information so that $\m{p}$ asymptotically converges to the set of formations which  {are} bearing congruent to $\m{p}^*$.
\end{problem}

\begin{problem} \label{problem:2}
Suppose that each agent can sense the bearing vectors with regard to its neighbors. Design control law for each agent using the available information so that $\m{p}$ converges to the set of formations which  {are} bearing congruent to $\m{p}^*$.
\end{problem}
\section{Adaptive bearing-based formation control with disturbance rejection}
In this section, we first propose a control law for Problem~\ref{problem:1}. Second, we analyse the effect of disturbance on the formation acquisition process, describe the time-varying target formation via the set of desired bearing constraints and the mismatched disturbance. Convergence of the moving target formation is then established via Barbalat's lemma. Finally, we propose a smooth adaptive formation control law that ensures global uniform ultimate boundedness of the desired moving formation.
\label{sec:3}
\subsection{Proposed control law}
We propose the following adaptive control law to solve Problem~\ref{problem:1}:
\begin{subequations}
\label{eq:bearing_based_control_law}
\begin{align} 
    \m{u}_i &= - \sum_{j\in \mc{N}_i} \gamma_{ij} \m{P}_{\m{g}^*_{ij}} \text{sgn}(\m{q}_{ij}), \, i=1,\ldots, n, \label{eq:bearing_based_control_law_a}\\
        \m{q}_{ij} &= {\m{P}_{\m{g}^*_{ij}}(\m{p}_i-\m{p}_j)}, \label{eq:bearing_based_control_law_b}\\
    \dot{\gamma}_{ij} &=k_{\gamma} \left|\left|\m{q}_{ij} \right|\right|_1,\, \forall (i,j) \in \mc{E}, \label{eq:bearing_based_control_law_c}
\end{align}
\end{subequations}
where $\m{P}_{\m{g}^*_{ij}}=\m{I}_d-\m{g}^*_{ij}(\m{g}^*_{ij})^\top$ can be determined from the desired bearing vector $\m{g}^*_{ij}$, $\text{sgn}(\cdot)$ is the signum function defined element-wise for each element of a vector, $\gamma_{ij}$ are adaptive gains satisfying $\gamma_{ij}(0)>0$, and $k_{\gamma}>0$ is a positive constant.

Let $\m{u}=[\m{u}_1^\top,\ldots,\m{u}_n^\top]^\top$, $\bm{\gamma}=[\ldots,\gamma_{ij},\ldots]^\top =[\gamma_1,\ldots,\gamma_m]^\top$, $\bm{\Gamma} = \text{diag}(\bm{\gamma})$, and $\bar{\bm{\Gamma}} = {\bm{\Gamma}} \otimes \m{I}_d$. The multiagent system under the proposed control law \eqref{eq:bearing_based_control_law} can be expressed in the following form:
\begin{subequations}
\label{eq:bearing_based_system}
\begin{align}
    \dot{\m{p}} &=- (\tilde{\m{R}}_{\rm b}^*)^\top \bar{\bm{\Gamma}}\text{sgn} \left(\tilde{\m{R}}_{\rm b}^*\m{p}\right)+\m{d}, \label{eq:bearing_based_system-1} \\
    \dot{\boldsymbol{\gamma}} &= k_{\gamma} \left[\|\m{q}_1\|_1,\ldots,\|\m{q}_m\|_1\right]^\top, \label{eq:bearing_based_system-2}
\end{align}
\end{subequations}
It is noted that the right-hand-side of Eq.~\eqref{eq:bearing_based_system-1} is discontinuous and we understand the solution of \eqref{eq:bearing_based_system} in Filippov sense \cite{Shevitz1994}. Furthermore, \eqref{eq:bearing_based_system-2} implies that the adaptive gains $\gamma_{ij}$ are non-decreasing in time.

\subsection{Effects of the matched and mismatched disturbances}
Let $\bar{\m{p}}(t) = \frac{1}{n}\sum_{i=1}^n\m{p}_i = \frac{1}{n}(\m{1}_n^\top\otimes\m{I}_d)\m{p}$ be the formation's centroid at time $t$. Let $\m{d}(t)=\m{d}_1 + \m{d}_2 + \m{d}_3$, where $\m{d}_1(t)\in \text{im}(\tilde{\m{R}}_{\rm b}(\m{p}^*))$ is a matched disturbance which can be handled by the control input $\m{u}$, $\m{d}_2(t)\in \text{im}(\m{1}_n \otimes \m{I}_d)$ is the disturbance part corresponding to a translation of the whole formation, and $\m{d}_3(t) \in \text{ker}(\m{L}_{\rm b})\setminus\text{im}(\m{1}_n\otimes\m{I}_d)$ is the disturbance part corresponding to a formation scaling about the formation's centroid. It is noticed that $\text{dim}(\text{ker}(\m{L}_{\rm b})\setminus\text{im}(\m{1}_n\otimes\m{I}_d)) = 1$.

Under the presence of the mismatched disturbances, the sliding-mode control law \eqref{eq:bearing_based_control_law} has no effect outside $\text{im}(\tilde{\m{R}}_{\rm b})$. The formation's uncontrollable space is driven by the disturbance parts $\m{d}_2$ and $\m{d}_3$, respectively. Defining  
\begin{align} \label{eq:vector_r}
    \m{r}(\m{p}) \triangleq \m{r} =\m{p}-\m{1}_n\otimes \bar{\m{p}},
\end{align}
and let $\m{0}_{dn} \neq \m{r}^*\in \text{ker}(\m{L}_{\rm b})\setminus\text{im}(\m{1}_n\otimes\m{I}_d)$ be a nonzero constant vector. Without loss of generality, $\m{r}^*$ is selected to be $\m{r}^* = \m{p}^*-\m{1}_n\otimes \bar{\m{p}}^*$, where $\bar{\m{p}}^*=\frac{1}{n}(\m{1}_n^\top \otimes \m{I}_d)\m{p}^*$.

Then, three components of the disturbance can be expressed as follows
\begin{subequations} \label{eq:disturbance}
\begin{align}
    \m{d}_1 &= (\tilde{\m{R}}_{\rm b}^*)^\top\m{f}_1(t), \\
    \m{d}_2 &= \m{1}_n\otimes \m{f}_2(t), \\
    \m{d}_3 &= f_3(t) \hat{\m{r}}^*,
\end{align}
\end{subequations}
where $\m{f}_1(t)=[\m{f}_{11}^\top,\ldots,\m{f}_{1m}^\top]^\top \in \mb{R}^{dm}$, $\m{f}_{1k}\in \mb{R}^d, \forall k=1,\ldots,m$, $\m{f}_2(t) \in \mb{R}^d$, $\hat{\m{r}}^* \triangleq \frac{\m{r}^*}{\|\m{r}^*\|}$ and ${f}_3(t) \in \mb{R}$. It is not hard to verify that $\m{d}_1,\m{d}_2,\m{d}_3$ are mutually orthogonal to each other. 

The following additional assumption is imposed on the disturbance.
\begin{assumption} \label{assumption:disturbance_term_linear}
The function $f_3(t)$ satisfies $\int_0^t|{f}_3(t)|dt < |\m{r}^\top(0)\hat{\m{r}}^*|$, for all $t \geq 0$.
\end{assumption}

Due to the assumption on the boundedness of $\m{d}$, the disturbance terms $\m{f}_i$ are bounded and uniformly continuous. We can thus assume that there exist $\beta_i > 0$ such that $\sup_{t\geq 0}\|\m{f}_i\|_{\infty} <\beta_i$, $\forall i=1,2,3$. The restricted effect of $f_3(t)$ in Assumption~\ref{assumption:disturbance_term_linear} prevents $n$ agents  {converging} into a point due to the scaling motion caused by disturbance.

We can now prove the following lemma on the effects of the disturbance to the formation. 

\begin{lemma} \label{lem:3.1}
Consider Problem \ref{problem:1} and suppose that Assumption~\ref{assumption:disturbance_term_linear} holds. Under the control law~\eqref{eq:bearing_based_control_law}, the following claims hold.
\begin{itemize}
\item[i.]  The formation's centroid moves with the velocity $\m{f}_2(t)$.
\item[ii.] The formation's scale ${s}(t)=\m{r}^\top(t)\hat{\m{r}}^*$ satisfies $|s(t)|>0$, $\forall t \geq 0$.
\end{itemize}
\end{lemma}

\begin{IEEEproof}
i. For the first claim, we have
\begin{align}
    \dot{\bar{\m{p}}} &= \frac{1}{n}(\m{1}_n^\top\otimes\m{I}_d) \dot{\m{p}} \nonumber\\
    &= -\frac{1}{n}(\m{1}_n^\top\otimes\m{I}_d) (\tilde{\m{R}}_{\rm b}^*)^\top\left( \bar{\bm{\Gamma}}\text{sgn} \left(\tilde{\m{R}}_{\rm b}^*\m{p}\right) - \m{f}_1 \right) \nonumber\\
    &\qquad + \frac{1}{n}(\m{1}_n^\top\m{1}_n\otimes\m{I}_d)\m{f}_2 + \frac{1}{n}(\m{1}_n^\top\otimes\m{I}_d) {f}_3(t) \frac{\m{r}^*}{\|\m{r}^*\|} \nonumber\\
                      &= \m{f}_2. \label{eq:centroid-bearing-based}
\end{align}
since ${\rm im}(\m{1}_n\otimes\m{I}_d)\in {\rm ker}(\tilde{\m{R}}_{\rm b}^*)$ and $\m{r}^* \perp {\rm im}(\m{1}_n\otimes \m{I}_d)$.

ii. Clearly, $\dot{s}=\dot{\m{r}}^\top\m{r}^*=f_3(t)$. Thus, $||s(t)|-|s(0)|| \le |s(t) - s(0)| \le \int_0^t|f_3(\tau)|d\tau < |s(0)|$ and it follows that $|s(t)|>0,\forall t \geq 0$.
\end{IEEEproof}

It is noted that  {the condition} $|s(t)|>0$ is necessary but sufficient  {to preventing collisions} among agents in the formation. We will consider collision avoidance among $n$ agents after showing convergence of $\m{p}$ to the set of target formation $\mc{D}$.

\subsection{Stability analysis}
From the analysis in the previous section, we can describe a time-varying target formation $\m{p}^d(t)$ as follows
\begin{itemize}
    \item Centroid: $\bar{\m{p}}^d(t) = \bar{\m{p}}(t)$,
    \item Scale: $s^d = s(t)$,
    \item Bearing constraints: $\m{P}_{\m{g}_{ij}^*}(\m{p}_i^d-\m{p}_j^d)=\m{0}_d$, $\forall (i,j)\in \mc{E}$.
\end{itemize}
Thus, 
\begin{align} \label{eq:BB_targetFormation}
    \m{p}^d \triangleq \m{1}_n\otimes \bar{\m{p}}(t) + \sqrt{n} s^d(t) \frac{\m{r}^*}{\|\m{r}^*\|}
\end{align}
is the time-varying target formation satisfying three aforementioned constraints. It follows from Lemma~\ref{lem:3.1} that $\m{p}^d \in \mc{D}\setminus {\rm im}(\m{1}_n), \forall t \geq 0$.

\begin{lemma} \label{lem:ker_Delta}
Suppose that Assumption~\ref{assumption:disturbance_term_linear} holds, the formation error $\bm{\delta} \triangleq \m{p}-\m{p}^d$ satisfies $\bm{\delta}(t) \perp {\rm ker}(\m{L}_{\rm b}^*)$, for all time $t\geq 0$. 
\end{lemma}
\begin{IEEEproof}
We have
\begin{align}
    \dot{\bm{\delta}} &= - (\tilde{\m{R}}_{\rm b}^*)^\top \bar{\bm{\Gamma}}\text{sgn} \left(\tilde{\m{R}}_{\rm b}^*\m{p}\right)+\m{d}_1  \nonumber \\
    &=- (\tilde{\m{R}}_{\rm b}^*)^\top \left(\bar{\bm{\Gamma}} \text{sgn} \left(\tilde{\m{R}}_{\rm b}^*\m{p}\right)+\m{f}_1\right). \label{eq:delta_displacement}
\end{align}
Thus, $\dot{\bm{\delta}} \perp {\rm ker}(\tilde{\m{R}}_{\rm b}^*) = {\rm ker}(\m{L}_{\rm b}^*)$, which implies that $[\m{1}_n\otimes \m{I}_d,\hat{\m{r}}^*]^\top\bm{\delta}(t)=[\m{1}_n\otimes \m{I}_d,\hat{\m{r}}^*]^\top\bm{\delta}(0)= \m{0}_{d+1}$. It follows that $\bm{\delta}(t) \perp {\rm ker}(\m{L}_b^*)$, for all $t\geq 0$.
\end{IEEEproof}

Using the fact that $\m{P}_{\m{g}^*}\bar{\m{H}}\bm{\delta}=\m{P}_{\m{g}^*}\bar{\m{H}}\m{p}-\m{P}_{\m{g}^*}\bar{\m{H}}\m{p}^d=\m{P}_{\m{g}^*}\m{z}$, equation \eqref{eq:delta_displacement} can be rewritten as a $\bm{\delta}$-dynamics follows
\begin{align}
    \dot{\bm{\delta}} &= - (\tilde{\m{R}}_{\rm b}^*)^\top \left(\bar{\bm{\Gamma}} \text{sgn} \left(\tilde{\m{R}}_{\rm b}^*\bm{\delta}\right)+\m{f}_1\right). \label{eq:delta_displacement1}
\end{align}

The main result of this section is stated in the following theorem.
\begin{theorem}[Stability of the desired formation] \label{thm:1}
Consider Problem \ref{problem:1} and suppose that Assumption~\ref{assumption:disturbance_term_linear} is satisfied. Then, the following claims hold.
\begin{itemize}
    \item[i.] If there exists $T\ge 0$ such that $\gamma_{ij}(T)>\beta_1\forall (i,j)\in \mc{E}$, the system \eqref{eq:delta_displacement1} is globally finite-time stable.
    \item[ii.] In general case, $\bm{\delta}(t) \to \m{0}_{dn}$ as $t\to +\infty$.
\end{itemize}
In both cases, there exists a constant vector $\bm{\gamma}^*\in \mb{R}^m$ such that $\bm{\gamma}\to \bm{\gamma}^*$, as $t\to+\infty$.
\end{theorem}

\begin{IEEEproof} i. For the first claim, consider the Lyapunov function $V = \frac{1}{2} \|\boldsymbol{\delta}\|^2$, which is continuously differentiable, positive definite, radially unbounded, and being bounded by two class $\mc{K}_{\infty}$ functions $\frac{h_1}{2} \|\boldsymbol{\delta}\|^2$ and $\frac{h_2}{2} \|\boldsymbol{\delta}\|^2$, for any $0<h_1 < 1 < h_2$. Then, we can compute the generalized gradient of $V$ \cite{Shevitz1994} for $t\geq T$ as $\partial V = \{\boldsymbol{\delta}\}$. The time derivative of $V$ is then
\[\dot{V} \in^{a.e} \dot{\tilde{V}} = \bigcup_{\boldsymbol{\xi} \in \partial V} \boldsymbol{\xi}^\top \text{K}[\dot{\boldsymbol{\delta}}],\] 
where ${\rm K}[{\rm sgn}](\cdot)$ denotes the differential inclusion, which is defined as
\begin{align}
    {\rm K}[{\rm sgn}](x) = \left\lbrace \begin{array}{rl}
        1, &  x>0,\\
       -1,  & x<0,\\
       {[-1,1]}, & x=0,
    \end{array} \right.
\end{align}
for a real number $x$. Further, ${\rm K}[{\rm sgn}](\cdot)$ is defined element-wise for a real vector. It follows that
\begin{align*}
    \dot{{V}} &= \boldsymbol{\delta}^\top (\tilde{\m{R}}_{\rm b}^*)^\top \left(-\bar{\boldsymbol{\Gamma}}\text{K}[\text{sgn}] \big(\m{P}_{\m{g}^*}\m{z}\big)+\m{f}_1\right) \\
    &= \m{z}^\top \m{P}_{\m{g}^*} \left(-\bar{\boldsymbol{\Gamma}}\text{K}[\text{sgn}] \big(\m{P}_{\m{g}^*}\m{z}\big)+\m{f}_1\right) \\
    &= -\sum_{k=1}^m \left(\gamma_{k} \|\m{P}_{\m{g}_{k}^*} \m{z}_{k} \|_1 - \m{z}_{k}^\top \m{P}_{\m{g}_{k}^*} \m{f}_{1k} \right) \\
    &\leq -\sum_{k=1}^m \left( \gamma_{k} - \|\m{f}_{1k}\|_{\infty} \right)  \|\m{P}_{\m{g}_{k}^*} \m{z}_{k} \|_1.
\end{align*}
Let $\zeta = \min_{k=1,\ldots,m} (\gamma_{k}(T) - \beta_1) =\min_{k=1,\ldots,m} (\gamma_k(T) - \beta_1) > 0$, we have
\begin{align}
    \dot{V} &\leq - \zeta  \sum_{k=1}^m \|\m{P}_{\m{g}_k^*} \m{z}_k \|_1 \leq -\zeta \|\m{P}_{\m{g}^*}\m{z}\|_1 \leq -\zeta \| \tilde{\m{R}}_{\rm b}^* \m{p}\|_1 \nonumber\\
    &\leq -\zeta \| \tilde{\m{R}}_{\rm b}^* \m{p}\| \leq - \zeta \left(\boldsymbol{\delta}^\top \boldsymbol{\m{L}}_{\rm b}^* \boldsymbol{\delta}\right)^{1/2}. \label{eq:5}
\end{align}
Now, using the fact that $\bm{\delta}(t) \perp {\rm ker}(\m{L}_{\rm b}^*)$, we have
\begin{align}
    \dot{V} &\leq -\zeta \left(\lambda_{d+1}(\m{L}_{\rm b}^*)\bm{\delta}^\top\bm{\delta}\right)^{1/2} \nonumber \\ 
    &\leq -\zeta \sqrt{\lambda_{d+1}(\m{L}_{\rm b}^*)}\|\bm{\delta}\| \nonumber \\ 
    &\leq -\underbrace{\zeta\sqrt{2\lambda_{d+1}(\m{L}_{\rm b}^*)}}_{\triangleq \rho}V^{1/2}.
\end{align}
Consider the differential equation $\dot{y}=-\rho y^{1/2}$, $y(t)=V(t)\geq 0$, $t\geq T$. Separating the variables gives $\frac{dy}{\sqrt{y}} = -\rho dt$, and thus, the solution can be found as $\sqrt{y(t)}=\sqrt{y(T)} -\frac{1}{2}\varepsilon t$ for $T\leq t \leq T_0$ and $y(t)=0$ for $t\ge T_0 \triangleq \frac{2\sqrt{y(0)}}{\varepsilon}$. Based on comparison lemma \cite[Section 3.4]{Khalil2002nonlinear}, $V(t)\leq y(t),\forall t\geq T$. Thus, $V(t)=0$ and $\bm{\delta}=\m{0}_{dn}$ for $t\ge T_0$. 

Since $\bm{\delta}=\m{0}_{dn}$ implies that $\m{q}_{ij}=\m{0}_d$, $\dot{{\gamma}}_{ij}=\m{0}$ and thus $\gamma_{ij}=\gamma_{ij}(T_0)$, $\forall t\geq T_0$.

ii. Consider the Lyapunov function $V = \frac{1}{2} \|\bm{\delta}\|^2 + \frac{1}{2k_{\gamma}} \|\bm{\gamma} - \gamma^* \m{1}_n\|^2$, for some $\gamma^* > \beta_1$. Then, $V$ is continuously differentiable, positive definite and radially unbounded with regard to $\bm{\eta}=[\bm{\delta}^\top,(\bm{\gamma} - \gamma^* \m{1}_n)^\top]^\top$. 

Similar to the proof of Lemma~\ref{lem:3.1}, we can compute
\begin{align}
    \dot{V} &= -\sum_{k=1}^m \left(\gamma_k \|\m{P}_{\m{g}_k^*} \m{z}_k \|_1 - \m{z}_k^\top \m{P}_{\m{g}_k^*} \m{f}_k \right)  \nonumber\\
    &\qquad\qquad + \sum_{k=1}^m (\gamma_k - \gamma^*) \|\m{P}_{\m{g}_k^*} \m{z}_k \|_1   \nonumber\\
    &\leq -\sum_{k=1}^m (\gamma^* - \beta_1) \|\m{P}_{\m{g}_k^*} \m{z}_k \|_1   \nonumber\\
    &\leq -(\gamma^* - \beta_1)\sum_{k=1}^m \|\m{P}_{\m{g}_k^*} \m{z}_k \|  \nonumber\\
    &\leq - (\gamma^* - \beta_1) \left( \bm{\delta}^\top \m{L}_{\rm b}^* \bm{\delta}\right)^{1/2}  \nonumber\\
    &\leq -(\gamma^* - \beta_1)\sqrt{\lambda_{d+1}(\m{L}_{\rm b}^*)}\|\bm{\delta}\| \leq 0,
\end{align}
which implies that $\bm{\delta}$, $\bm{\gamma}-\bm{\gamma}^*$ are uniformly bounded, and so is $\bm{\gamma}$. Moreover, there exists $\lim_{t \to + \infty} V(t) \geq 0$. Since $\dot{V}$ is uniformly continuous, it follows from the Barbalat's lemma \cite{Khalil2002nonlinear} that $\dot{V} \to 0$, as $t \to +\infty$. Because $\bm{\delta} \perp \text{ker}({\m{L}}_{\rm b}^*)$, it follows that $\bm{\delta} \to \m{0}_{dn}$, as $t \to +\infty$. 

Finally, as ${\gamma}_{ij}$ are bounded and non-decreasing, there must exist finite constants $\gamma_{ij}^*$ such that $\gamma_{ij}\to \gamma_{ij}^*$, $\forall (i,j)\in \mc{E}$.
\end{IEEEproof}

From Theorem~\ref{thm:1}, we conclude that $\m{p}(t) \to \mc{D}$ as desired. Now, we impose further requirements for collision avoidance. 

\begin{corollary}[Collision avoidance 1] \label{cor:CA1} Suppose that Assumption~\ref{assumption:disturbance_term_linear} holds, and 
\begin{align} \label{eq:collisionAvoidance}
\inf_{t\geq 0}\|\m{p}_{i}^d-\m{p}_{j}^d\|-\sqrt{2nV(0)}>0,\; \forall i,j\in\mc{V},\; i\neq j,
\end{align}
there will be no collision between $n$ agents for all time $t\geq 0$.
\end{corollary}
\begin{IEEEproof}
Consider the inequality
\begin{align*}
\|\m{p}_i^d - \m{p}_j^d\| &= \|(\m{p}_i^d - \m{p}_i) - (\m{p}_j^d - \m{p}_j) + (\m{p}_i-\m{p}_j)\| \\
& \leq \|\m{p}_i - \m{p}_i^d\| + \|\m{p}_j - \m{p}_j^d\| + \|\m{p}_i-\m{p}_j\|,\\
&\leq \sqrt{n\sum_{i=1}^n\|\m{p}_i-\m{p}^d_i\|^2}+\|\m{p}_i-\m{p}_j\|\\
&=\sqrt{n}\|\bm{\delta}\|+\|\m{p}_i-\m{p}_j\|,
\end{align*}
From the analysis in Theorem~\ref{thm:1}, there holds $\|\bm{\delta}(t)\| \leq \sqrt{2V(t)}\leq \sqrt{2V(0)}$. By combining $\|\m{p}_i-\m{p}_j\| \geq \|\m{p}_i^d - \m{p}_j^d\| - \sqrt{2nV(0)} \geq \inf_{t\geq 0}\|\m{z}_{ij}^d\|-\sqrt{2nV(0)} $ with the inequality \eqref{eq:collisionAvoidance}, the distance between any pair of agents $i,j\in \mc{V},i\neq j$ maintains positive.
\end{IEEEproof}

\subsection{Bearing-only stabilization with time-dependent adaptive control gains}
Although the formation control law \eqref{eq:bearing_based_control_law_a} is expressed in the same form as a matrix-weighted consensus protocol \cite{Trinh2026Springer}[Chapter 16], it can be implemented using only bearing measurements since 
${\rm sgn}(\m{P}_{\m{g}_{ij}^*}(\m{p}_{i}-\m{p}_j))={\rm sgn}(\m{P}_{\m{g}_{ij}^*}\m{g}_{ij})$. Thus, if the magnitude of the disturbance is known, \eqref{eq:bearing_based_control_law_a} provides a bearing-only stabilization control law. However, the adaptive mechanism \eqref{eq:bearing_based_control_law_c} uses displacements (in $\|\m{q}_{ij}\|$) for adjusting $\gamma_{ij}$. To this end, we propose the following bearing-only adaptive control law
\begin{subequations}
\label{eq:bearing_based_control_law1}
\begin{align} 
\m{u}_i &= - \sum_{j\in \mc{N}_i} \gamma_{ij} \m{P}_{\m{g}^*_{ij}} \text{sgn}(\m{P}_{\m{g}^*_{ij}}\m{g}_{ij}), \, i=1,\ldots, n, \label{eq:bearing_based_control_law1_a}\\
{\gamma}_{ij} &= {\gamma}_{ij}^0 + t{\gamma}_{ij}^1{\rm sgn}\left(\left|\left|\m{P}_{\m{g}^*_{ij}}\m{g}_{ij} \right|\right|_1\right),  \label{eq:bearing_based_control_law1_b} \\ 
\dot{\gamma}_{ij}^0 &=k_{\gamma} \left|\left|\m{P}_{\m{g}^*_{ij}}\m{g}_{ij} \right|\right|_1,~ \dot{\gamma}_{ij}^1 =k_{\gamma}\left|\left|\m{P}_{\m{g}^*_{ij}}\m{g}_{ij} \right|\right|_1 t, \label{eq:bearing_based_control_law1_c}
\end{align}
\end{subequations}
with $~{\gamma}_{ij}^0(0)>0,{\gamma}_{ij}^1(0)>0,\forall (i,j) \in \mc{E}$. We prove the following lemma.

\begin{lemma} \label{lem:3.3} Suppose that Assumption~\ref{assumption:disturbance_term_linear} is satisfied, under the control law~\eqref{eq:bearing_based_control_law1}, there exist $C_0, C_1 >0$ such that $\|\m{z}(t)\|\leq C_0 + C_1 t$, for all $t\geq 0$.
\end{lemma}
\begin{IEEEproof}
It  {suffices} to establish the upper bound in case the adaptive gains are constants and initialized at $\gamma_{ij}(0)< \|\m{f}_{1k}\|,~\forall k=1,\ldots,m$ and $\forall t\ge 0$, as the control gains are non-decreasing under the adaptive laws \eqref{eq:bearing_based_control_law1_c}. 

Consider the function $V = \frac{1}{2}\|\bm{\delta}\|^2$, and following the same steps in Theorem~\ref{thm:1}, one has
$\dot{V} = \sum_{k=1}^m (\|\m{f}_{1k}\|-\gamma_{k})\|\m{P}_{\m{g}_{ij}^*}\m{z}_{ij}\|_1 \leq \xi \sum_{(i,j)\in\mc{E}}\|\m{P}_{\m{g}_{ij}^*}\m{z}_{ij}\|_1$, where $\xi = \sup_{t\ge0}\max_{k=1,\ldots,m} |\|\m{f}_{1k}\|-\gamma_{k}|>0$. Then,
$\dot{V}\leq \xi\|\m{P}_{\m{g}^*}\m{z}\|_1\le \|\tilde{\m{R}}_{\rm b}^*\m{p}\|\leq \xi (\bm{\delta}^\top\m{L}_{\rm b}^*\bm{\delta})^{1/2}\le \xi \sqrt{2\lambda_{dn}(\m{L}_{\rm b}^*)}V^{1/2}$. By comparison lemma,
$V(t)\leq y(t)$, where $y(t)$ is the solution of $\dot{y}=\rho_1\sqrt{y},~y(0)=V(0)$ and $\rho_1 = \xi \sqrt{2\lambda_{dn}(\m{L}_{\rm b}^*)}$. Hence,
\begin{align}
\sqrt{V(t)} &\leq \sqrt{V(0)} + \rho_1 t,\nonumber \\
\|\m{p}-\m{p}^d\| &\leq \sqrt{2V(0)}+\sqrt{2}\rho_1 t \nonumber\\
\|\m{z}-\m{z}^d\|&=\|\bar{\m{H}}(\m{p}-\m{p}^d)\| \leq \|\bar{\m{H}}\|\|\m{p}-\m{p}^d\| \nonumber \\
&\leq \sqrt{2V(0)}\|\bar{\m{H}}\|+\sqrt{2}\rho_1\|\bar{\m{H}}\| t,~\forall t\geq 0. \label{eq:bearing_based_control_law1_d}
\end{align}
Meanwhile, from Lemma~\ref{lem:3.1} and \eqref{eq:BB_targetFormation}, it follows that
\begin{align}
\|\m{z}^d\|=\|\bar{\m{H}}\m{p}^d\|=\sqrt{n} s(t)\frac{\|\m{z}^*\|}{\|\m{r}^*\|} \leq 2\sqrt{n} s(0)\frac{\|\m{z}^*\|}{\|\m{r}^*\|}, \label{eq:bearing_based_control_law1_e}
\end{align}
or i.e., $\|\m{z}^d\|$ is upper bounded by a constant. It follows from \eqref{eq:bearing_based_control_law1_d}--\eqref{eq:bearing_based_control_law1_e} that there exist $C_0,C_1>0$ such that $\|\m{z}(t)\|\leq C_0 + C_1 t$ for $t\geq 0$.
\end{IEEEproof}

It follows immediately from Lemma~\ref{lem:3.3} and $\|\m{z}\|=\sqrt{\sum_{k=1}^m\|\m{z}_k\|^2}$ that for each $e_k\equiv (i,j)\in \mc{E}$, there exist $C_{0,k},C_{1k}>0$ such that $\|\m{z}_{ij}\|\equiv \|\m{z}_{k}\| \leq C_{0k}+C_{1k}t,~t\geq 0$. Intuitively, the adaptation rule \eqref{eq:bearing_based_control_law1_c} tunes two components $\gamma_{ij}^0,\gamma_{ij}^1$ accordingly to eventually dominate the unknown coefficients $C_{0k},C_{1k}$.

\begin{theorem} \label{thm:1_1}
Consider Problem \ref{problem:1}. Suppose that Assumption~\ref{assumption:disturbance_term_linear} is satisfied and there exists $\chi>0$ such that $\|\m{z}_{ij}\|>\chi$ for all $t\ge 0$. Then, under the control law~\eqref{eq:bearing_based_control_law1}, $\bm{\delta}(t) \to \m{0}_{dn}$, and there exists a constant vector $\bm{\gamma}^*\in \mb{R}^m$ such that $\bm{\gamma}\to \bm{\gamma}^*$ in finite time. 
\end{theorem}

\begin{IEEEproof}
Consider the Lyapunov function $V = \frac{1}{2} \|\bm{\delta}\|^2 + \sum_{k=1}^m\frac{\chi}{2k_{\gamma}} \|{\gamma}_k^0 - \gamma^{0*}\|^2 + \sum_{k=1}^m\frac{\chi}{2k_{\gamma}} \|{\gamma}_k^1 - \gamma^{1*}\|^2$, for some $\gamma^{l*} > \frac{\beta_1}{\chi} \max_{k=1,\ldots,m}C_{kl},~l=0,1$. Then, $V$ is continuously differentiable, positive definite and radially unbounded with regard to $\bm{\eta}=[\bm{\delta}^\top,(\bm{\gamma}^0 - \gamma^{0*} \m{1}_n)^\top,(\bm{\gamma}^1 - \gamma^{1*} \m{1}_n)^\top]^\top$. Moreover, by denoting $\tilde{\m{q}}_{k}=\m{P}_{\m{g}_k^*} \m{g}_k$, we have $\m{P}_{\m{g}_k^*} \m{z}_k=\|\m{z}_k\| \tilde{\m{q}}_{k}$ and
\begin{align}
\dot{V} &= -\sum_{k=1}^m \left(\gamma_k \|\m{P}_{\m{g}_k^*} \m{z}_k \|_1 - \m{z}_k^\top \m{P}_{\m{g}_k^*} \m{f}_k \right)  \nonumber\\
    &~ +\chi\sum_{k=1}^m (\gamma_{k}^0+\gamma_k^1 t)\|\tilde{\m{q}}_{k}\|_1  - \chi\sum_{k=1}^m (\gamma^{0*}+\gamma^{1*} t)\|\tilde{\m{q}}_{k} \|_1   \nonumber\\
   &\leq -\sum_{k=1}^m \gamma_k\|\tilde{\m{q}}_{k}\|_1 (\|\m{z}_k\|-\chi) \nonumber \\
    &\quad - \sum_{k=1}^m\left(\chi\gamma^{0*} - \beta_1C_{0k} +t (\chi\gamma^{1*} - C_{1k}\beta_1) \right)\|\tilde{\m{q}}_{k}\|_1. \label{eq:bearing_based_control_law1_f}
\end{align}
From the definitions of $\gamma^{l*},l=0,1,~C_{k0},~C_{k1}$ and $\beta_1$, it follows that $\dot{V}\leq 0$, and $\bm{\eta}$ is uniformly bounded. This implies $\|\bm{\delta}\| \leq \|\bm{\delta}(0)\|$, and there exists $M>0$ such that $\sup_{t\geq 0}\max_{k=1,\ldots,m}\|\m{z}_{k}(t)\|\leq M$. Let $\xi = \frac{1}{M} \min_{k}\left(\chi \gamma^{0*} - \beta_1C_{0k}\right)$, then following a similar analysis as in Theorem~\ref{thm:1}, equation \eqref{eq:bearing_based_control_law1_f} implies that
\begin{align}
\dot{V}\leq - \xi \sum_{k=1}^m \|\m{P}_{\m{g}_k^*}\m{z}_k\| \leq -\chi\sqrt{\lambda_{d+1}(\m{L}_{\rm b}^*)}\|\bm{\delta}\|. \label{eq:bearing_based_control_law1_g}
\end{align}
Since $\dot{V}$ is uniformly continuous, by Barbalat's lemma $\|\bm{\delta}(t)\|\to 0$, and $\lim_{t\to+\infty}\gamma_{ij}^0$, $\lim_{t\to+\infty}\gamma_{ij}^1(t)$ exist and are finite.

Now, we prove that $\lim_{t\to+\infty}\gamma_{ij}$ exist and are finite by contradiction. Suppose the contradiction that $\gamma_{ij}$ grow unbounded, then there exists $T_1>0$ such that $\gamma_{ij}(t)>\beta_1$ for all $(i,j)\in \mc{E}$ and $t\ge T_1$. By Theorem~\ref{thm:1}(i), there exists $T\geq T_1$ such that $\|\bm{\delta}\|=0$ for $t\geq T$. This implies $\gamma_{ij}=\gamma_{ij}^0$ for $t \geq T$. However, $\lim_{t\to+\infty} \gamma_{ij}^0$ exist and are finite, and thus $\gamma_{ij}$ is bounded. The existence and finiteness of $\lim_{t\to+\infty}\gamma_{ij}(t)$ follows from this contradiction. Since $\gamma_{ij}(t)$ have finite limits, it follows from \eqref{eq:bearing_based_control_law1_b} that $\|\m{P}_{\m{g}_{ij}^*}\m{g}_{ij}\|=0$ in finite time. Notice that $\dot{\gamma}_{ij}^l=0$ whenever $\|\m{P}_{\m{g}_{ij}^*}\m{g}_{ij}\|=0$, we conclude that both $\bm{\delta} = 0$  and $\gamma_{ij} = \gamma_{ij}^*$ in finite time.
\end{IEEEproof}

\begin{remark}
It is remarked that using explicit time in control law may be problematic, as $\gamma_{ij}$ grow rapidly before convergent if $\m{d}_{ij}$ attains its peak lately. Since bounded input is a typical requirement for any realistic control system, we need to establish boundedness of $\gamma_{ij}$. To this end, we have to assume the existence of $\chi>0$, which can be arbitrarily small, such that $\|\m{z}_{ij}\|>\chi,\forall t\geq 0$. Theorem~\ref{thm:1_1} merely states that for all trajectories that no collision happens, the desired formation can be achieved in finite time. Conservative conditions can be derived based on initial positions of the agents and the rate of changes of the disturbances, but they are not very useful, as the upper bound of the disturbances are known in the design.

With certain modifications, the proposed time-dependent adaptive design method in this subsection is applicable to a leaderless consensus \cite{Cao2011TAC}, leader-follower and acyclic leader-follower bearing-only formation tracking problems with unknown leaders' velocity in \cite{Nguyen2024TCyber} and \cite{Trinh2021LCSS}.
\end{remark}

\subsection{A smooth bearing-based control law}
A main issue in implementing sliding-mode control laws is chattering, which is due to the discontinuities of the signum function. An effective strategy for avoiding chattering is compromising the control performance, for example, relaxing the requirement for perfect disturbance rejection. To this end, the following adaptive bearing-based control law is proposed for each agent
\begin{subequations}
\label{eq:Practical_Bearing_B}
\begin{align} 
    \m{u}_i &= -k_p \sum_{j\in \mc{N}_i} \m{q}_{ij} - \sum_{j\in \mc{N}_i} \gamma_{ij} \frac{\m{q}_{ij}}{\|\m{q}_{ij}\|+\varepsilon}, \label{eq:PractBearingB1}\\
    \dot{\gamma}_{ij} &= k_{\gamma} \left(\frac{\|\m{q}_{ij}\|^2}{\| \m{q}_{ij}\|+\varepsilon} -\alpha\gamma_{ij}\right), \; \forall j\in\mc{N}_i,\label{eq:PractBearingB2}
\end{align}
\end{subequations}
where $\m{q}_{ij}=\m{P}_{\m{g}_{ij}^*} (\m{p}_{i}-\m{p}_j)$, $\varepsilon>0$, and $\alpha>0$ are small positive constants, and $k_p>0$ is a positive control gain. It is remarked that a related adaptive strategy was introduced in \cite{Mei2020} for a consensus problem.

Similar to the previous subsection, we can prove that if Assumption~\ref{assumption:disturbance_term_linear} holds, under the smooth control law \eqref{eq:Practical_Bearing_B}, both claims of Lemma~\ref{lem:3.1} are still valid. Therefore, we can also define the time-varying target formation $\m{p}^d$ as in Eq.~\eqref{eq:BB_targetFormation}. Let $\bm{\delta}=\m{p}-\m{p}^d$, and $\tilde{\m{P}}_{\rm g^*}^\epsilon \triangleq {\rm blkdiag}\left(\frac{{\m{P}}_{\rm g^*_k}}{\|\m{q}_{k}\|+\varepsilon}\right)$, the $\bm{\delta}$-dynamics can be written as follows
\begin{align}
    \dot{\bm{\delta}} &=- k_p\m{L}_{\rm b}^*\bm{\delta} - \bar{\m{H}}^\top \tilde{\m{P}}_{\rm g^*}\left(\bar{\bm{\Gamma}} \tilde{\m{P}}_{\rm g^*}^\epsilon \bar{\m{H}}\bm{\delta}+\m{f}_1\right). \label{eq:delta_displacementP}
\end{align}

The main result of this subsection is stated in the following theorem.
\begin{theorem}[Globally uniform ultimate boundedness of the desired formation] \label{thm:3.2} Consider Problem~\ref{problem:1} and suppose that Assumption \ref{assumption:disturbance_term_linear} holds. Under the control law \eqref{eq:Practical_Bearing_OnlyC}, $\bm{\delta}$ and $\bm{\gamma}$ are uniformly ultimately bounded, and the ultimate bound is adjustable by the control law's parameters $k_p,k_{\gamma},\varepsilon,$ and $\alpha$.
\end{theorem}
\begin{IEEEproof}
 The proof can be found in Appendix \ref{app:thm3_2}.
\end{IEEEproof}

\section{Adaptive bearing-only formation control with disturbance rejection}
\label{sec:4}
This section focuses on Problem~\ref{problem:2}, i.e., designing bearing-only formation control for $n$ agents with disturbances. The complexity of the problem is significantly increased, as the agents need to simultaneously achieve all the desired bearing constraints and suppress the disturbances at the same time relying on only the measured bearing vectors. We propose an adaptive sliding-mode based bearing-only control law, and then provide a corresponding stability analysis. Then, we provide an analysis on a smooth adaptive bearing-only control law, which guarantees uniform ultimate boundedness of the desired formation.  {The bearing-only control laws in this section can be considered as the formation control law in \cite{zhao2015bearing} equipped with new adaptive mechanisms.}

\subsection{Proposed control law}
Consider the system of single-integrator agents with uncertainty \eqref{eq:agent_model}. The bearing-only control law for each agent $i =1,\ldots,n$ is given as follows
\begin{subequations}
\label{eq:Bearing_OnlyC}
\begin{align} 
    \m{u}_i &= -\sum_{j\in \mc{N}_i} \gamma_{ij} \m{P}_{\m{g}_{ij}} \text{sgn}\left(\m{P}_{\m{g}_{ij}} \m{g}_{ij}^* \right), \label{eq:Bearing_OnlyC1}\\
    \dot{\gamma}_{ij} &= k_{\gamma} \|\m{P}_{\m{g}_{ij}} \m{g}_{ij}^*\|_1,\; \gamma_{ij}(0)>0,\; \forall j\in \mc{N}_i. \label{eq:Bearing_OnlyC2}
\end{align}
\end{subequations}
Unlike the control law \eqref{eq:bearing_based_control_law1} in the previous section, the adaptive rule in \eqref{eq:Bearing_OnlyC2} is time-independent. 
Since the right-hand-side of Eq.~\eqref{eq:system_BOM} is discontinuous, the solution of \eqref{eq:system_BOM} is understood in Filippov sense. Clearly, $\gamma_{ij}(t) \geq \gamma_{ij}(0)>0$ for all $t\geq 0$ and $(i,j)\in \mc{E}$.

\subsection{Stability of formations under matched disturbances}
In this subsection, the stability analysis of the $n$-agent system under the bearing-only control law \eqref{eq:Bearing_OnlyC} will be considered. We will first consider the following assumption on the disturbance in the stability analysis, as the target formation in this case is well-defined. 
\begin{assumption} \label{assumption:2} The disturbance $\m{d}$ is uniformly continuous and bounded. Moreover $\m{d} \perp \m{p}$ 
for all time $t\geq 0$.
\end{assumption}

Under Assumption~\ref{assumption:2}, the disturbance can be expressed as $\m{d} = \tilde{\m{R}}_{\rm b}^\top \m{f}_1+\m{1}_n\otimes \m{f}_2$, where the upper bounds $\sup_{t\ge 0} \|\m{f}_1\|_{\infty} = \beta_5>0$ and $\sup_{t\ge 0} \|\m{f}_2\|_{\infty} = \beta_2>0$ are finite. By denoting $\m{P}_{\m{g}}={\rm blkdiag}(\m{P}_{\m{g}_k})$, we can express the $n$-agent system under the control law \eqref{eq:Bearing_OnlyC1} in the matrix form as follows
\begin{align}
    \dot{\m{p}} &= \tilde{\m{R}}_{\rm b}^\top \left( \bar{\bm{\Gamma}}  \text{sgn}\left(\m{P}_{\m{g}} \m{g}^* \right) + \m{f}_1\right) + \m{1}_n\otimes \m{f}_2\nonumber\\
    &= \tilde{\m{R}}_{\rm b}^\top \left( \bar{\bm{\Gamma}} \text{sgn} \left( \text{blkdiag} (\|\m{z}_k^d\| \m{P}_{\m{g}_k}) \m{g}^* \right) + \m{f}_1\right) + \m{1}_n\otimes \m{f}_2\nonumber\\
    &= \tilde{\m{R}}_{\rm b}^\top \left(\bar{\bm{\Gamma}} \text{sgn} \left(\m{P}_{\m{g}} \m{z}^d \right)+ \m{f}_1\right) + \m{1}_n\otimes \m{f}_2\nonumber\\
    &= -\tilde{\m{R}}_{\rm b}^\top \left(\bar{\bm{\Gamma}} \text{sgn} \left( \m{P}_{\m{g}}(\m{z}-\m{z}^d) \right) - \m{f}_1\right)+ \m{1}_n\otimes \m{f}_2 \nonumber\\
    &= -\tilde{\m{R}}_{\rm b}^\top \left(\bar{\bm{\Gamma}} \text{sgn} \left( \m{P}_{\m{g}} \bar{\m{H}} (\m{p}-\m{p}^d) \right) - \m{f}_1\right) + \m{1}_n\otimes \m{f}_2. \label{eq:system_BOM}
\end{align}
Similarly, the vector form of \eqref{eq:Bearing_OnlyC2} is given as
\begin{align}
    \dot{\bm{\gamma}} = k_{\gamma} [\|\m{P}_{\m{g}_{1}} \m{g}_{1}^*\|_1,\ldots,\|\m{P}_{\m{g}_{m}} \m{g}_{m}^*\|_1]^\top, \label{eq:system_BOM1}
\end{align}
where $\bm{\gamma} = [\ldots,\gamma_{ij},\ldots]^\top = [\gamma_1,\ldots,\gamma_{m}]^\top$.

We have the following lemmas on \eqref{eq:system_BOM}.

\begin{lemma} \label{lem:4.1}
Consider the Problem \ref{problem:2} and suppose that Assumption~\ref{assumption:2} holds. Under the control law \eqref{eq:Bearing_OnlyC}, the following claims hold.
\begin{enumerate}
    \item[i.] The formation's centroid $\bar{\m{p}}=\frac{1}{n}(\m{1}_n^\top\otimes \m{I}_d)\m{p}$ moves at velocity $\m{f}_2$ and the formation's scale $s(t) \triangleq \frac{1}{\sqrt{n}}\|\m{r}(t)\|$ is time-invariant.
    \item[ii.] $\|\m{z}_{ij}\|=\|\m{p}_i - \m{p}_j\| \leq 2s\sqrt{n-1},~\forall i, j = 1,\ldots, n$, $\forall t\ge 0$.
    \item[iii.] Let $\m{p}^d,\m{p}^u \in \mb{R}^{dn}$ be two configurations that satisfy 
\begin{subequations}
\begin{align*}
    \m{p}^d &= \m{1}_n \otimes \bar{\m{p}} + \|\m{r}(0)\|\frac{\m{r}^*}{\|\m{r}^*\|}= \m{1}_n \otimes \bar{\m{p}} + \sqrt{n}s\frac{\m{r}^*}{\|\m{r}^*\|}, \\
    \m{p}^u &= \m{1}_n \otimes \bar{\m{p}} - \|\m{r}(0)\|\frac{\m{r}^*}{\|\m{r}^*\|}= \m{1}_n \otimes \bar{\m{p}} - \sqrt{n}s\frac{\m{r}^*}{\|\m{r}^*\|}.
\end{align*}    
\end{subequations}
Then, we have
\begin{enumerate}
\item Formation's centroid: $\bar{\m{p}}^d = \bar{\m{p}}^u= \bar{\m{p}}(t)$,
\item Formation's scale: $s = s^d= \frac{1}{\sqrt{n}}\|\m{r}^d\|= \frac{1}{\sqrt{n}}\|\m{r}(0)\| = s(0) =\frac{1}{\sqrt{n}}\|\m{r}^u\| = s^u$, where $\m{r}^d=\m{r}(\m{p}^d)$ and $\m{r}^u=\m{r}(\m{p}^u)$
\item Bearing vectors: $\frac{\m{p}_j^d - \m{p}_i^d}{\|\m{p}_j^d - \m{p}_i^d\|} = \m{g}_{ij}^*$, and $\frac{\m{p}_j^u - \m{p}_i^u}{\|\m{p}_j^u - \m{p}_i^u\|} = -\m{g}_{ij}^*$, $\forall i,j =1, \ldots, n$, $i \neq j$.
\end{enumerate}
\end{enumerate}
\end{lemma}
\begin{IEEEproof} 
i. Consider the system~\eqref{eq:system_BOM}, we have
\begin{align*}
    \dot{\bar{\m{p}}} &= \frac{1}{n}(\m{1}_n^\top\otimes\m{I}_d) \tilde{\m{R}}_{\rm b}\left(\bar{\bm{\Gamma}} {\rm sgn} \left( \m{P}_{\m{g}} \bar{\m{H}} (\m{p}-\m{p}^d) \right) - \m{f}_1\right) \\
    &\qquad + \frac{1}{n}(\m{1}_n^\top\otimes\m{I}_d)(\m{1}_n\otimes\m{f}_2)= \m{f}_2,\\
    \dot{s} &= \frac{\m{r}^\top\left(\dot{\m{p}} - \m{1}_n\otimes \dot{\bar{\m{p}}} \right)}{\sqrt{n}\|\m{r}\|}=0
\end{align*}
because $\m{r}\perp {\rm im}(\tilde{\m{R}}_{\rm b})$ and $\m{r}\perp {\rm im}(\m{1}_n\otimes \m{I}_d)$.

ii. This proof is similar to \cite{zhao2015bearing}[Corollary 2]. We have $\|\m{p}_i-\bar{\m{p}}\|^2 = \|\sum_{j=1,j\neq i}^n(\m{p}_j-\bar{\m{p}})\|^2 \leq (n-1)\sum_{j=1,j\ne i}^n\|\m{p}_j-\bar{\m{p}}\|^2$, which implies that $\|\m{p}_i-\bar{\m{p}}\|^2 \leq \frac{n-1}{n}\|\m{p}-\m{1}_n\otimes\bar{\m{p}}\|^2$. Thus, $\|\m{p}_i-\bar{\m{p}}\| \leq \sqrt{n-1}s$.

Since $\|\m{p}_i-\m{p}_j\|\leq \|\m{p}_i - \bar{\m{p}}\| + \|\m{p}_j - \bar{\m{p}}\|$, we have $\|\m{p}_i-\m{p}_j\|\leq \|\m{p}_i - \bar{\m{p}}\| + \|\m{p}_j - \bar{\m{p}}\| \leq 2\sqrt{n-1}s$.

iii. It is not hard to verify that 
\begin{align*}
    \bar{\m{p}}^d &=\frac{1}{n}(\m{1}_n^\top\otimes\m{I}_d)\m{p}^d=\bar{\m{p}}(t),\\
    \m{r}^d &= \m{p}^d - \m{1}_n\otimes \bar{\m{p}}^d = \|\m{r}(0)\|\frac{\m{r}^*}{\|\m{r}^*\|},\\
    s^d &=\|\m{r}^d\| = \left|\left|\|\m{r}(0)\| \frac{\m{r}^*}{\|\m{r}^*\|} \right|\right| = \|\m{r}(0)\|=s(0),
\end{align*}
and a similar formula can be derived for $\bar{\m{p}}^u$ and $s^u$.
\end{IEEEproof}

The following lemma proves finite-time stability of the target formation under Assumption~\ref{assumption:2}.
\begin{lemma} \label{lem:4.2}
Consider Problem \ref{problem:2} and suppose that Assumption~\ref{assumption:2} holds. Under the control law~\eqref{eq:Bearing_OnlyC}, if $\gamma_{ij}(0)>\beta_5, \forall (i,j) \in \mc{E}$, the desired configuration $\m{p}=\m{p}^d$ is finite time stable.
\end{lemma}
\begin{IEEEproof}
Let $\bm{\delta} = \m{p} - \m{p}^d = \m{r} - \m{r}^d$, it follows from \eqref{eq:system_BOM} that
\begin{align}
  \dot{\bm{\delta}}  \in -\tilde{\m{R}}_b^\top \left(\bar{\bm{\Gamma}} \text{K[sgn]} \left( \m{P}_{\m{g}} \bar{\m{H}} \bm{\delta} \right) - \m{f}_1\right). \label{eq:system_BOM2}
\end{align}
The Lyapunov function $V = \frac{1}{2}\|\bm{\delta}\|^2$ is continuously differentiable, positive definite and radially unbounded, and has
\begin{align}
    \dot{V} \in^{\text{a.e}} \dot{\tilde{V}} &= -\bm{\delta}^\top \bar{\m{H}}^\top \m{P}_{\m{g}} \left(\tilde{\bm{\Gamma}} \text{K[sgn]} \left(\m{P}_{\m{g}}\bar{\m{H}} \bm{\delta}\right) - \m{f}_1 \right) \nonumber\\
    &= - \sum_{k=1}^m \left(\gamma_{k} \|\m{P}_{\m{g}_k}\m{z}^d\|_1 - (\m{z}^d_k)^\top \m{P}_{\m{g}_k} \m{f}_{1k} \right) \nonumber \\
    &\leq - \sum_{k=1}^m (\gamma_{k}(0) - \|\m{f}_{1k}\|_{\infty}) \|\m{P}_{\m{g}_k}\m{z}_k^d\|_1.
\end{align}
Let $\zeta = \min_{k\in \{1,\ldots,m\}} (\gamma_{k}(0) - \beta_5)>0$, then we have
\begin{align}
\dot{V}  \leq  -\zeta\sum_{k=1}^m \|\m{P}_{\m{g}_k}\m{z}^d_k\|_1 
\leq -\zeta\|\m{P}_{\m{g}}\m{z}^d\|_1 \leq -\zeta\|\m{P}_{\m{g}}\m{z}^d\|_2. \label{eq:dotV_BOM}
\end{align}
By noting that $\m{z}^d_k = \|\m{z}^d_k\|\m{g}_{k}^*$, we can write
\begin{align}
 \|\m{P}_{\m{g}}\m{z}^d\|_2^2&=(\m{z}^d)^\top \m{P}_{\m{g}}\m{z}^d \nonumber\\
 &=\sum_{k=1}^m (\m{z}^d_k)^\top \m{P}_{\m{g}_k} \m{z}^d_k \nonumber \\
 &=\sum_{k=1}^m \|\m{z}^d_k\|^2 (\m{g}^*_k)^\top \m{P}_{\m{g}_k} \m{g}^*_k. \label{eq:gPg}
\end{align}
Moreover, it follows from $(\m{g}^*_{ij})^\top \m{P}_{\m{g}_{ij}}\m{g}^*_{ij} = (\m{g}^*_{ij})^\top (\m{I}_d - \m{g}_{ij}\m{g}_{ij}^\top)\m{g}^*_{ij} = 1-(\m{g}_{ij}^\top)\m{g}^*_{ij})^2 = (\m{g}_{ij})^\top (\m{I}_d - \m{g}_{ij}^*\m{g}_{ij}^{*\top})\m{g}_{ij}=(\m{g}_{ij})^\top \m{P}_{\m{g}_{ij}^*}\m{g}_{ij}$, equation \eqref{eq:dotV_BOM} and inequality \eqref{eq:gPg}  that 
\begin{align}
   \dot{V} &\leq -\zeta \left(\sum_{k=1}^m \frac{\|\m{z}_k^d\|^2}{\|\m{z}_k\|^2} \m{z}_k^\top \m{P}_{\m{g}_k^*} \m{z}_k \right)^{\frac{1}{2}}\nonumber\\
    &\leq -\underbrace{\zeta \left(\frac{\min_{k}\|\m{z}_k^d\|^2}{4(n-1)s^2}\right)^{\frac{1}{2}}}_{\triangleq \varsigma} \left(\sum_{k=1}^m \m{z}_k^\top \m{P}_{\m{g}_k^*} \m{z}_k\right)^{\frac{1}{2}} \nonumber\\
    &\leq -\varsigma \left(\m{z}^\top \m{P}_{\m{g}^*} \m{z} \right)^{\frac{1}{2}} \nonumber\\
    &\leq -\varsigma \left(\bm{\delta}^\top\bar{\m{H}}^\top \m{P}_{\m{g}^*} \bar{\m{H}} \bm{\delta} \right)^{\frac{1}{2}}\nonumber\\
    &\leq -\varsigma \left(\bm{\delta}^\top\m{L}_{\rm b}^*\bm{\delta}\right)^{\frac{1}{2}} \leq 0.
\end{align}
Furthermore, $\dot{V} = 0$ if and only if $\bm{\delta} \in \text{ker}(\m{L}_b(\m{p}^*))$. It follows from Lemma~\ref{lem:4.1} that $\m{p}= \m{p}^d$ or  $\m{p}=\m{p}^u$. 

To establish instability of $\m{p}^u$ under the proposed control law \eqref{eq:Bearing_OnlyC}, we consider the Lyapunov function $W = \frac{1}{2}\|\bm{\delta}'\|^2$, where $\bm{\delta}' \triangleq\m{p} - \m{p}^u$ and the $\bm{\delta}'$-dynamics. Since ${\rm sgn}(\m{P}_{\m{g}}\bar{\m{H}} \bm{\delta}')=-{\rm sgn}(\m{P}_{\m{g}}\bar{\m{H}} \m{p}^u)=-{\rm sgn}(\m{P}_{\m{g}}\m{z}^u)={\rm sgn}(\m{P}_{\m{g}}\m{z}^d)={\rm sgn}(\m{P}_{\m{g}}\m{g}^*)$, we can write
\begin{align}
    \dot{W} &= -(\bm{\delta}')^\top\bar{\m{H}}^\top \m{P}_{\m{g}} \left(-\tilde{\bm{\Gamma}} \text{K[sgn]} \left(\m{P}_{\m{g}}\bar{\m{H}} \bm{\delta}'\right) - \m{f}_{1} \right) \nonumber\\
    &= \sum_{k=1}^m (\gamma_{ij} \|\m{P}_{\m{g}_k}\m{z}^u_k\|_1 - \m{f}_{1k}^\top\m{P}_{\m{g}_k}\m{z}^u_k) \nonumber\\
    &\geq \sum_{k=1}^m (\gamma_{ij}-\beta_5) \|\m{P}_{\m{g}_k}\m{z}^u_k\|_1 \geq 0. 
\end{align}
For any open ball centered at $\m{0}_{dn}$ which does not contain $\m{p}^d$ and $\bm{\delta}'=\m{p}^d-\m{p}^u$, $\dot{W}>0$. Based on Chetaev's instability theorem \cite{Khalil2002nonlinear}, we conclude that $\m{p}^u$ is unstable.

Due to the existence of disturbance, the system cannot stay identically at $\m{p}^u$. Thus, there exists a finite time $t_1\ge 0$ such that $\m{p}(t_1) \neq \m{p}^u$. Thus, for $t \ge t_1$, $\dot{V}$ is negative definite, which implies that $\m{p}(t) \to\m{p}^d$ asymptotically. 

Decomposing $\bm{\delta} = \bm{\delta}^{||}+ \bm{\delta}^{\perp}$, where $\bm{\delta}^{||} \in \text{ker}(\m{L}_{\rm b}^*)$, and $\bm{\delta}^{\perp} \in \text{im}(\m{L}_{\rm b}^*)$ as illustrated in Fig.~\ref{fig:geometry}. Since $\|\bm{\delta} + \m{r}^d\|= \|\m{r}^d\|=\sqrt{n} s(0)$ is time-invariant, let $\alpha(t) \in [0, \frac{\pi}{2})$ denote the angle between $\bm{\delta}(t)$ and $-\m{r}^d$, for $t\geq t_1$, it follows that 
\begin{align*}
    \bm{\delta}^\top \m{L}_b^* \bm{\delta} &= (\bm{\delta}^\perp)^\top \m{L}_{\rm b}^* (\bm{\delta}^\perp) \\
    & \geq \lambda_{d+2}(\m{L}_{\rm b}^*) \|\bm{\delta}^\perp\|^2 \\
    & \geq \lambda_{d+2}(\m{L}_{\rm b}^*)\sin^2\alpha \|\bm{\delta}\|^2.
\end{align*}
As $\m{p}(t_1) \neq \m{p}^u$, and $\m{p}$ asymptotically converges to $\m{p}^d$, for $t\ge t_1$, we have $\alpha$ is non-decreasing. Thus,
\begin{align}
    \dot{V} \leq -\underbrace{\varsigma \sin(\alpha(t_1))\sqrt{ 2\lambda_{d+2}(\m{L}_{\rm b}^*)}}_{\triangleq \varrho } \frac{\|\bm{\delta}\|}{\sqrt{2}} \leq -\varrho  V^{\frac{1}{2}},
\end{align}
which implies that $V(t) \to 0$ in finite time by a similar argument as in the proof of Theorem~\ref{thm:1}.
\end{IEEEproof}

The main result of this section is given in the following theorem. 
\begin{theorem} \label{thm:2}
Consider the Problem \ref{problem:2} and suppose that Assumption~\ref{assumption:2} holds. Under the adaptive bearing-only control law~\eqref{eq:Bearing_OnlyC}, the following claims hold.
\begin{enumerate}
    \item[i.] $\m{p}(t) \to \m{p}^d$ as $t \to + \infty$, 
    \item[ii.] There exists a constant vector $\bm{\gamma}^*$ such that $\bm{\gamma} \to \bm{\gamma}^*$, as $t \to + \infty$,
    \item[iii.] Additionally, if $\gamma_i^*>\beta_5,~\forall i$, and there exists a finite time $T$ such that $|\gamma_i(T) - \gamma_i^*|< \min_{i=1,\ldots,m}|\gamma_i(T) - \beta_5|,~\forall i=1,\ldots,n,$ then $\m{p}(t) \to \m{p}^d$ in finite time.
\end{enumerate}
\end{theorem}
\begin{IEEEproof}
i. Consider the Lyapunov function $V = \frac{1}{2} \|\bm{\delta}\|^2 + \frac{1}{2k} \sum_{k=1}^m\|\m{z}_k^d\| \|{\gamma_k} - \beta_6 \m{1}_n\|^2$, for some $\beta_6>\beta_5$. Similar to Lemma~\ref{lem:4.2}, we have
\begin{align}
\dot{V} \in^{\text{a.e}} \dot{\tilde{V}} &= -\bm{\delta}^\top \bar{\m{H}}^\top \m{P}_{\m{g}} \left(\tilde{\bm{\Gamma}} \text{K[sgn]} \left(\m{P}_{\m{g}} \bar{\m{H}} \bm{\delta}\right) - \m{f}_1 \right) \nonumber\\
    &\qquad\qquad + \sum_{k=1}^m \|\m{z}^d_k\|(\gamma_k - \beta_5) \|\m{P}_{\m{g}_k} \m{g}_k^* \|_1 \nonumber\\
    &= - \sum_{k=1}^m \left(\gamma_{k} \|\m{P}_{\m{g}_k}\m{z}^d_k\|_1 - (\m{z}^d_k)^\top \m{P}_{\m{g}_k} \m{f}_{1k} \right)  \nonumber \\
    &\qquad\qquad + \sum_{k=1}^m (\gamma_k - \beta_6) \|\m{P}_{\m{g}_k} \m{z}_k^d \|_1 \label{eq:ddotV_BOM} \\
    &\leq - \sum_{k=1}^m \left(\beta_{6}  - \|\m{f}_{1k}\|_{\infty} \right) \|\m{P}_{\m{g}_k}\m{z}_k^d\|_1.\nonumber
\end{align}
For $\beta_6 > \beta_5 \ge  \sup_{t\ge 0}\|\m{f}_1\|_{\infty}$, and thus $\zeta = \beta_6 - \beta_5>0$. We obtain $\dot{V} \leq - \zeta \|\m{P}_{\m{g}_k}\m{z}_k^d\|_1$. Similar to the proof of Theorem~\ref{thm:2}, we can eventually prove that $\dot{V}\leq - \varphi \left(\bm{\delta}^\top\m{L}_{\rm b}^*\bm{\delta}\right)^{\frac{1}{2}} \leq 0$, for some $\varphi>0$. It follows that $\bm{\delta}$ and $\bm{\gamma}$ are uniformly bounded and there exists finite  $\lim_{t\to+\infty}V(t) \leq V(0)$. Due to the uniformly continuity of $\m{f}$, it follows from \eqref{eq:ddotV_BOM} that $\dot{V}$ is uniformly continuous. As assumptions of the Barbalat's lemma is satisfied, we conclude that $\dot{V}\to 0$. However, $\dot{V}=0$ if and only if $\m{p}=\m{p}^d$ or $\m{p}=\m{p}^u$. The instability of $\m{p}^u$ can be similarly proved as in Theorem~\ref{thm:2}, using the function $W=\frac{1}{2} \|\bm{\delta}'\|^2 + \frac{1}{2k_{\gamma}} \sum_{k=1}^m\|\m{z}_k^d\| \|{\gamma_k} - \beta_6 \m{1}_n\|^2$. Therefore, $\m{p}\to\m{p}^d$ as $t\to+\infty$.

ii. This claim follows from the fact that $\gamma_k,k=1,\ldots,m,$ are non-decreasing and upper bounded.

iii. If the assumptions of this claim hold, then $\gamma(T)\ge \beta_5$. The claim follows from Theorem~\ref{thm:2}.
\end{IEEEproof}

The collision avoidance condition in Corollary~\ref{cor:CA1} still holds, and the following corollary can be proved. 
\begin{corollary}[Collision avoidance 2]\label{cor:CA2} Under the control law \eqref{eq:Bearing_OnlyC}, if 
\begin{align} \label{eq:collisionAvoidance2}
\inf_{t\geq 0}\|\m{p}_i^d - \m{p}_j^d\| -  \sqrt{2nV(0)}>0,\; \forall i,j\in\mc{V},\; i\neq j,
\end{align}
no collision happens between $n$ agents for all time $t\geq 0$.
\end{corollary}
\begin{figure}
\centering
\includegraphics[width=.75\linewidth]{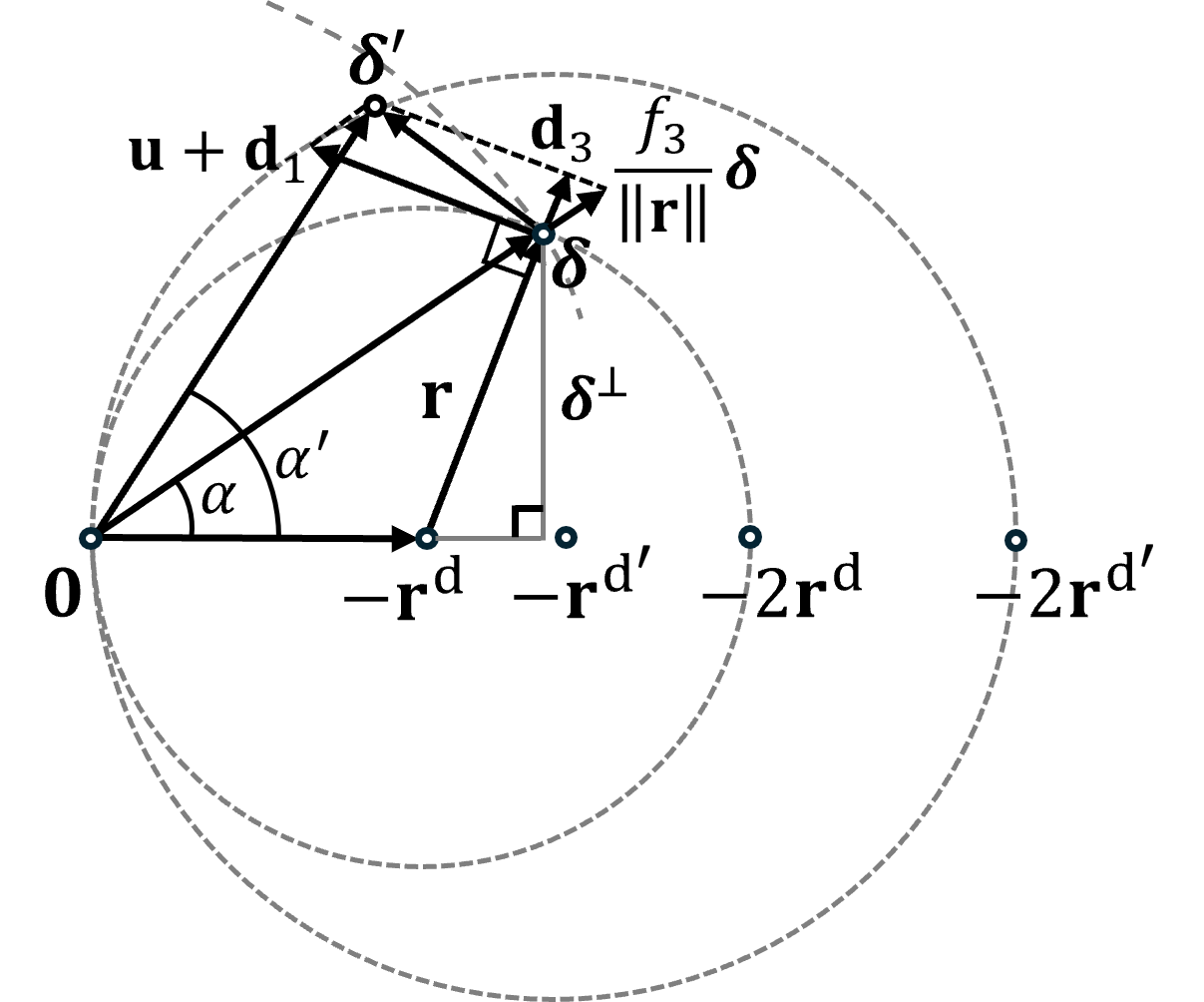}
\caption{The diagram describes the system \eqref{eq:dot_deltaBOM}. If the scaling disturbance component $\m{d}_3(t)\equiv \m{0}$, then the system is constrained on the sphere $\|\bm{\delta}+\m{r}^d\|=\|\m{r}^d\|=\|\m{r}^d(0)\|$. If the scaling disturbance $\m{d}_3$ is nonzero, the formation scale is not invariant, and $\|\bm{\delta}\|$ may increase or decrease. In both cases, $\alpha$ denotes the angle between $\bm{\delta}$ and $-\m{r}^*$.}
\label{fig:geometry}
\end{figure}

\subsection{Effects of scaling disturbance}
Let the disturbance $\m{d}$ consist an additional term $\m{d}_3 = f_3(t) \frac{\m{r}}{\|\m{r}\|}$ which rescales the whole formation. Then, the $n$-agent system under control law~\eqref{eq:Bearing_OnlyC} is
\begin{align}
\dot{\m{p}} = -\tilde{\m{R}}_{\rm b}^\top(\bm{\Gamma}{\rm sgn}(\tilde{\m{R}}_{\rm b}\bm{\delta})-\m{f}_1) + \m{1}_n\otimes\m{f}_2 + f_3\frac{\m{r}}{\|\m{r}\|}.
\end{align}
With $s=\frac{1}{\sqrt{n}}\|\m{r}\|$, we have 
\begin{align*}
    \dot{s} &= \frac{\m{r}^\top}{\sqrt{n}\|\m{r}\|}\left(-\tilde{\m{R}}_{\rm b}^\top(\bm{\Gamma}{\rm sgn}(\tilde{\m{R}}_{\rm b}\bm{\delta})-\m{f}_1) + f_3\frac{\m{r}}{\|\m{r}\|}\right)\\
    &=\frac{f_3\m{r}^\top\m{r}}{\sqrt{n}\|\m{r}\|^2}=\frac{f_3}{\sqrt{n}}.
\end{align*}
The desired formation is defined as
\begin{align}
    \m{p}^d = \m{1}_n \otimes \bar{\m{p}} + \|\m{r}\|\frac{\m{r}^*}{\|\m{r}^*\|}= \m{1}_n \otimes \bar{\m{p}} + \sqrt{n}s\frac{\m{r}^*}{\|\m{r}^*\|}.
\end{align}
Let $\bm{\delta} = \m{p}-\m{p}^d = \sqrt{n}s\left(\frac{\m{r}}{\|\m{r}\|} - \frac{\m{r}^*}{\|\m{r}^*\|}\right)=\m{r}-\|\m{r}\|\frac{\m{r}^*}{\|\m{r}^*\|}=\m{r}-\m{r}^d$, we have
\begin{align}
\dot{\bm{\delta}} &= -\tilde{\m{R}}_{\rm b}^\top\left(\bm{\Gamma}{\rm sgn}(\tilde{\m{R}}_{\rm b}\bm{\delta})-\m{f}_1\right) + f_3 \left(\frac{\m{r}}{\|\m{r}\|} - \frac{\m{r}^*}{\|\m{r}^*\|}\right)\nonumber \\
&= -\tilde{\m{R}}_{\rm b}^\top\left(\bm{\Gamma}{\rm sgn}(\tilde{\m{R}}_{\rm b}\bm{\delta})-\m{f}_1 \right) + \frac{f_3}{\|\m{r}\|}\bm{\delta}. \label{eq:dot_deltaBOM}
\end{align}

Fig.~\ref{fig:geometry} depicts an example where the scaling disturbance perturbs $\bm{\delta}$ to $\tilde{\bm{\delta}}=\bm{\delta}(t+dt)\approx \bm{\delta}(t)+ \frac{f_3}{\|\m{r}\|}\bm{\delta}dt$, with $\|\tilde{\bm{\delta}}\|>\|\bm{\delta}\|$ and $dt>0$ is small time increment. Due to the input $\m{u}+\m{d}_1$, $\bm{\delta}$ is actually controlled to $\bm{\delta}'=\bm{\delta}(t+dt)\approx \bm{\delta}(t)+ (\m{u}+\m{d}_1+\frac{f_3}{\|\m{r}\|}\bm{\delta}(t))dt$, with $\|\bm{\delta}'\|<\|\bm{\delta}\|$ and $\alpha'=\alpha(t+dt)>\alpha(t)$, where $\alpha(t)>0$ is the angle between $\bm{\delta}$ and $-\m{r}^d=-\sqrt{n}s(t)\m{r}^*$ at time $t$. 

For simplicity, suppose that $\gamma_{ij}(0)$ are sufficient large to suppress the disturbance term $\m{d}_1$, consider the Lyapunov function $V=\frac{1}{2}\|\bm{\delta}\|^2$, we have
\begin{align}
\dot{V} &= \bm{\delta}^\top{\rm K}[\dot{\bm{\delta}}]=\bm{\delta}^\top \left(-\tilde{\m{R}}_{\rm b}^\top\left(\bm{\Gamma}{\rm K}[{\rm sgn}](\tilde{\m{R}}_{\rm b}\bm{\delta})-\m{f}_1)\right) + \frac{f_3}{\|\m{r}\|}\bm{\delta} \right) \nonumber\\
    &\leq - \frac{\zeta\min_k\|\m{z}_k^d\|}{2\sqrt{n-1}s} \sqrt{\lambda_{d+2}(\m{L}_{\rm b}^*)}\|\bm{\delta}\| + \frac{1}{\sqrt{n}s}f_3\|\bm{\delta}\|^2, \label{eq:dotV_scaling}
\end{align}
where the derivation of immediate inequalities are similar to those in Lemma \ref{lem:4.2} and Lemma~\ref{lem:4.1} iii.b, and $\zeta$ is a positive constant. 

We have the following theorem.
\begin{theorem}\label{thm:2_1}
Consider the system~\eqref{eq:dot_deltaBOM}. Suppose that $s_m<|s(t)|<s_M$, $\alpha(0)>0$ and $\gamma_{ij}(0)$ are sufficient large so that
\begin{align} \label{eq:dotV_scaling2}
\sin\alpha(0) > \frac{4s_M\sqrt{n-1}\sup_{t\geq 0}|f_3|}{\zeta\sqrt{n \lambda_{d+2}(\m{L}^*_{\rm b})}\inf_{t\ge 0}\min_{k}\|\m{z}_k^d\|}
\end{align}
satisfies. Then, $\bm{\delta}\to\m{0}_{dn}$ in finite time.
\end{theorem}
\begin{IEEEproof}
From $\cos\alpha = -(\m{r}^d)^\top \bm{\delta}$, it follows that
\begin{align}
-\sin\alpha \dot{\alpha} &= -(\dot{\m{r}}^d)^\top\bm{\delta} - (\m{r}^d)^\top \dot{\bm{\delta}} \nonumber \\
\sin\alpha \dot{\alpha} &= \frac{f_3}{\|\m{r}\|}\|\bm{\delta}\|^2 - \bm{\delta}^\top\dot{\bm{\delta}} \nonumber \\
& =  \bm{\delta}^\top \tilde{\m{R}}_{\rm b}^\top\left(\bm{\Gamma}{\rm K}[{\rm sgn}](\tilde{\m{R}}_{\rm b}\bm{\delta})-\m{f}_1)\right) \nonumber\\
&=  \sum_{k=1}^m \left(\gamma_{k} \|\m{P}_{\m{g}_k}\m{z}^d\|_1 - (\m{z}^d_k)^\top \m{P}_{\m{g}_k} \m{f}_{1k} \right) \nonumber \\
\dot{\alpha} & \geq \sum_{k=1}^m (\gamma_{k}(0) - \|\m{f}_{1k}\|_{\infty}) \|\m{P}_{\m{g}_k}\m{z}_k^d\|_1 \geq 0.
\end{align}
Thus, $\alpha(0) \leq \alpha(t) \leq \frac{\pi}{2}~\forall t$. 

Noting that $\|\bm{\delta}\|\le \|-2\m{r}^d\| \le 2 s_M$. Thus, $\|\bm{\delta}\|^2 \leq 2 s_M\|\bm{\delta}\|$. Since \eqref{eq:dotV_scaling2} satisfies for $\alpha(0)\le \alpha(t)~\forall t$, \eqref{eq:dotV_scaling2} holds for all $t\geq 0$. Let \[\varepsilon \triangleq \frac{\zeta\inf_{t\ge 0}\min_k\|\m{z}_k^d\|}{2\sqrt{n-1}s_M}\lambda_{d+2}^{1/2}(\m{L}^*)\sin\alpha(0) - \frac{2\sup_{t\geq 0}|f_3|}{\sqrt{n}},\] it follows from \eqref{eq:dotV_scaling} and \eqref{eq:dotV_scaling2} that
\begin{align}\label{eq:dotV_scaling1}
\dot{V} \leq - \varepsilon \|\bm{\delta}\|,~\forall t\ge 0.
\end{align}
Therefore, $\|\bm{\delta}\| \to \m{0}_{dn}$, $\alpha \to \frac{\pi}{2}$ in finite time by Theorem~\ref{thm:2} iii.
\end{IEEEproof}

\begin{remark}
Theorem~\ref{thm:2_1} can be extended in case $\gamma_{ij}(0)$ is not large enough. Under the assumption that $s_m<|s(t)|<s_M$, $\alpha(0)>0$, by using the Lyapunov function $V = \frac{1}{2} \|\bm{\delta}\|^2 + \frac{1}{2k} \sum_{k=1}^m\|\m{z}_k^d\| \|{\gamma_k} - \beta_6 \m{1}_n\|^2$, for some $\beta_6$ large enough so that \eqref{eq:dotV_scaling2} holds, we also obtain \eqref{eq:dotV_scaling} and  \eqref{eq:dotV_scaling1} as in proof of Theorem~\ref{thm:2} i. and ii. Thus, $\bm{\delta}\to \m{0}$ and $\lim_{t\to+\infty}\gamma_{ij}(t)$ exist and are finite.
\end{remark}

\subsection{A smooth adaptive bearing-only control law}
A main issue in implementing sliding-mode control laws is the chattering phenomenon, which arises from the discontinuity of the signum function. To mitigate this effect, we modify the bearing-only control law \eqref{eq:Bearing_OnlyC} for each agent as follows
\begin{subequations}
\label{eq:Practical_Bearing_OnlyC}
\begin{align} 
    \m{u}_i &= -k_p \sum_{j\in \mc{N}_i}\m{P}_{\m{g}_{ij}} \m{g}_{ij}^* - \sum_{j\in \mc{N}_i} \gamma_{ij} \frac{\m{P}_{\m{g}_{ij}} \m{g}_{ij}^*}{\|\m{P}_{\m{g}_{ij}} \m{g}_{ij}^*\|+\varepsilon}, \label{eq:PractBearing_OnlyC1}\\
    \dot{\gamma}_{ij} &= k_{\gamma} \left(\frac{\|\m{P}_{\m{g}_{ij}} \m{g}_{ij}^*\|^2}{\|\m{P}_{\m{g}_{ij}} \m{g}_{ij}^*\|+\varepsilon} -\alpha\gamma_{ij}\right), \; \forall j\in\mc{N}_i,\label{eq:PractBearing_OnlyC2}
\end{align}
\end{subequations}
where $\varepsilon,\alpha$ are small positive constants, $k_p$, $k_{\gamma}$ are positive control gain and adaptive rate, respectively. {Given that no collision happens, the control law~\eqref{eq:Practical_Bearing_OnlyC} is smooth.}

We rewrite equation \eqref{eq:Practical_Bearing_OnlyC} governing the motion of $n$-agent system in matrix-form as follows
\begin{align}
    \dot{\m{p}} &= k_p\tilde{\m{R}}_{\rm b}^\top\m{g}^*+\bar{\m{H}}^\top \bar{\bm{\Gamma}} {\m{P}}_{\m{g}}^{\varepsilon}\m{g}^* + \m{d} \nonumber \\
    &= k_p\tilde{\m{R}}_{\rm b}^\top\m{g}^* + \bar{\m{H}}^\top {\m{P}}_{\m{g}}\bar{\bm{\Gamma}} {\m{P}}_{\m{g}}^{\varepsilon}\m{g}^* + \tilde{\m{R}}_{\rm b}^\top\m{f}_1 + \m{1}_n\otimes \m{f}_2\nonumber\\
    &= k_p\tilde{\m{R}}_{\rm b}^\top\m{g}^* + \tilde{\m{R}}_{\rm b}^\top\left(\bar{\bm{\Gamma}}{\m{P}}_{\m{g}}^{\varepsilon}\m{g}^*+\m{f}_1 \right) + \m{1}_n\otimes \m{f}_2, \label{eq:BOM_Pract_Syst}
\end{align}
where ${\m{P}}_{\m{g}}^{\varepsilon} \triangleq {\rm blkdiag}\left( \frac{\m{P}_{\m{g}_k}}{\|\m{P}_{\m{g}_k}\m{g}_k^*\|+\varepsilon}\right)$, and we have used the fact that $\m{P}_{\m{g}}\m{P}_{\m{g}}^{\varepsilon}=\m{P}_{\m{g}}^{\varepsilon}\m{P}_{\m{g}}=\m{P}_{\m{g}}^{\varepsilon}$. 

It is not hard to verify that $\dot{\m{p}}(t) \perp {\rm ker}(\tilde{\m{R}}_{\rm b})$, and thus all claims of Lemma~\ref{lem:4.1} hold with the system \eqref{eq:BOM_Pract_Syst}, \eqref{eq:PractBearing_OnlyC2}. Let $\bm{\delta}=\m{p}-\m{p}^d$, the main result of this subsection is provided in the following theorem.

\begin{theorem}[Uniform ultimate boundedness of the desired formation] \label{thm:4.2}
Suppose that the assumptions of Problem~\ref{problem:2} hold. Under the control law \eqref{eq:Practical_Bearing_OnlyC}, $\bm{\delta}$ and $\bm{\gamma}$ are globally uniformly ultimately bounded, and the ultimate bound is adjustable by the control law's parameters $k_p,k_{\gamma},\varepsilon,$ and $\alpha$.
\end{theorem}
\begin{IEEEproof}
The proof is provided in Appendix~\ref{app:thm4_2}.
\end{IEEEproof}

Finally, we remark that a similar collision avoidance condition in Corrolary~\ref{cor:CA2} holds for the system \eqref{eq:BOM_Pract_Syst}, \eqref{eq:PractBearing_OnlyC2}.

\section{Simulation results}
\label{sec:5}
In this section, we provide simulation examples to demonstrate the effectiveness of the proposed formation control laws proposed in Sections \ref{sec:3} and \ref{sec:4}. In each simulation, we consider a system of 6 agents in the two-dimensional space with the interaction graph $\mc{G}$ as depicted in Fig.~\ref{fig:graph}. The desired formation is specified by a set of 12 desired bearing vectors, corresponding to a six-regular polygon configuration $\m{p}^*$ with $\m{p}_1^* = [\frac{\sqrt{3}}{2},\frac{1}{2}]^\top$ and $\m{p}_{i+1}^*=\m{R}\left(\frac{\pi}{3}\right)$, for $i=2,\ldots,6$ and $\m{R}(\alpha)=\begin{bmatrix}
    \cos\alpha & -\sin\alpha \\ \sin\alpha  & \cos\alpha
\end{bmatrix}$. The initial values of agents' positions and adaptive gains in all simulations are chosen as in Table \ref{tab:my-table}.

\begin{table}[]
\centering
\caption{The initial positions $\m{p}_k(0)$, $k=1,\ldots,6,$ and adaptive gains $\gamma_k(0),k=1,\ldots,12$.}
\label{tab:my-table}
\begin{tabular}{@{}lcccccc@{}}
\toprule
$k$               & 1 & 2 & 3 & 4 & 5 & 6 \\ \midrule
$\m{p}_k$      & $\begin{bmatrix}{1.96}\\{-3.24}\end{bmatrix}$  & $\begin{bmatrix}{0.5}\\{2.35}\end{bmatrix}$  & $\begin{bmatrix}{-3.39}\\{-0.2}\end{bmatrix}$  & $\begin{bmatrix}{-5.6} \\ {-0.3}\end{bmatrix}$  & $\begin{bmatrix}{-2.0}\\{-5.3}\end{bmatrix}$  & $\begin{bmatrix}{3.8}\\ {-1.2}\end{bmatrix}$  \\
$\gamma_k$     &  0.13 &  0.51 & 0.18  & 0.79  &  0.85 &   0.5 \\
$\gamma_{6+k}$ & 0.85  & 0.08  & 0.51  & 0.07  &  0.43 &   0.1 \\ \bottomrule
\end{tabular}
\end{table}

\begin{figure}[t!]
\centering
\begin{tikzpicture}[scale=1,
    every node/.style={circle, draw, , inner sep=1pt}
]
\foreach \i in {1,...,6} {
    \node (v\i) at ({30+60*(\i-1)}:1) {\i};
}
\draw (v1) -- (v2) -- (v3) -- (v4) -- (v5) -- (v6) -- (v1);
\draw (v1) -- (v3) -- (v5) -- (v1);
\draw (v2) -- (v4) -- (v6) -- (v2);
\end{tikzpicture}
\caption{A regular graph $\mc{G}$ of six-vertices in simulation.}
\label{fig:graph}
\end{figure}
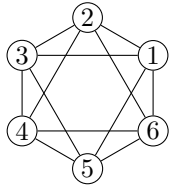

\subsection{Simulation 1: Bearing-based control}
First, we simulate the six-agent system under the control law~\eqref{eq:bearing_based_control_law}, with $k_{\gamma}=0.1$. The disturbance vector $\m{d}$ consists of three terms as described in \eqref{eq:disturbance}, with parameters
\begin{align*}
\m{d}_1 &= \frac{1}{4}\m{L}_{\rm b}^*(\m{1}_2\otimes [1,2,3,4,5,6,0,0.5,1,1.5,2,2.5]^\top), \\
\m{f}_2 &= \left[3+\cos(0.125\pi t),1.5\sin\left(\pi t/{6} \right)\right]^\top, \\
f_3 &= 0.125\left(8e^{-0.1t}+\sin\left({\pi t}/{10}\right) \right).
\end{align*}
\begin{figure}[t]
\centering
\subfloat[]{\includegraphics[width=.45\textwidth]{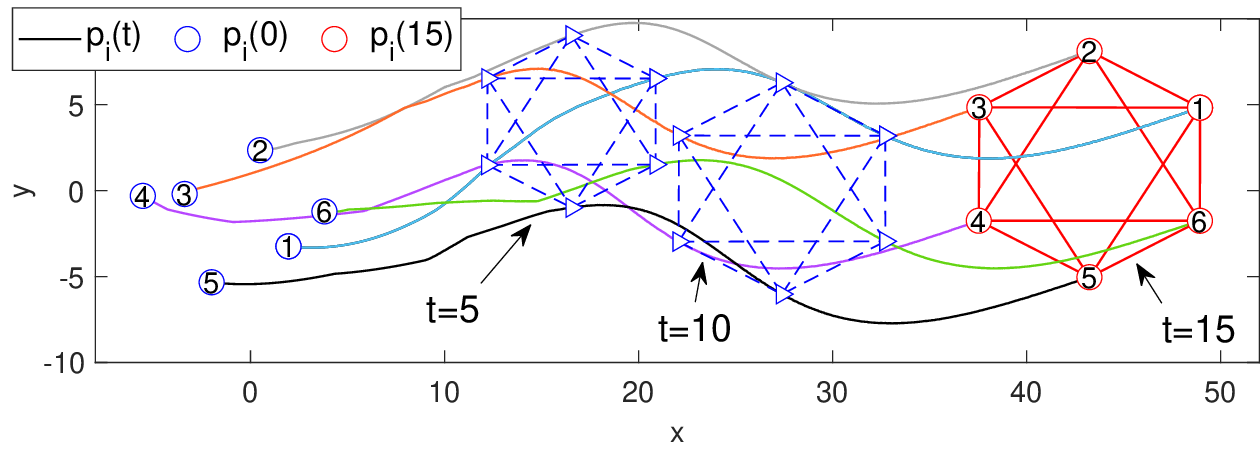}} \\
\subfloat[]{\includegraphics[width=.46\textwidth]{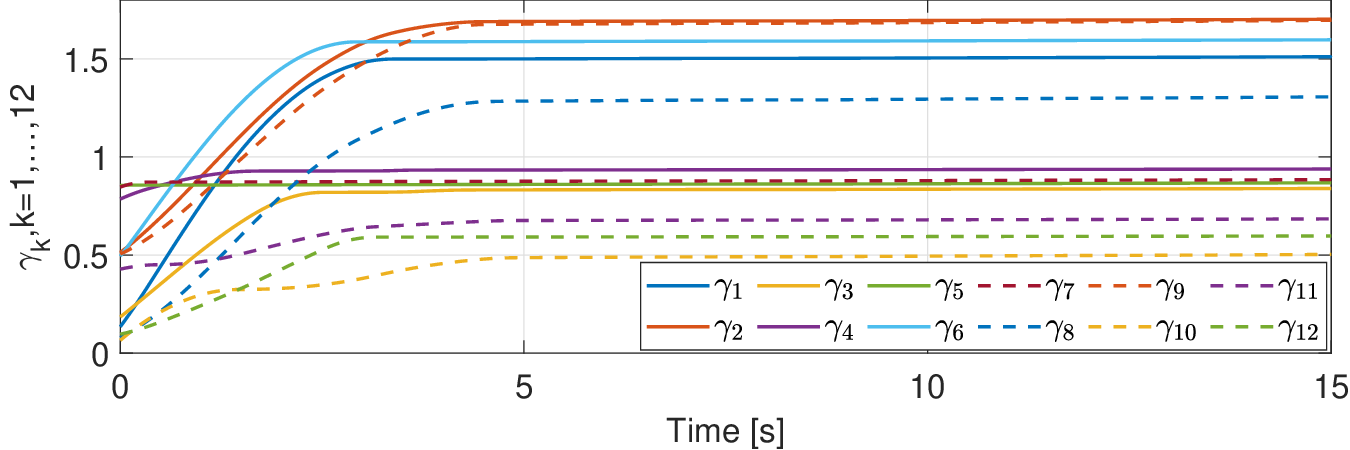}} \\
\subfloat[]{\includegraphics[width=.45\textwidth]{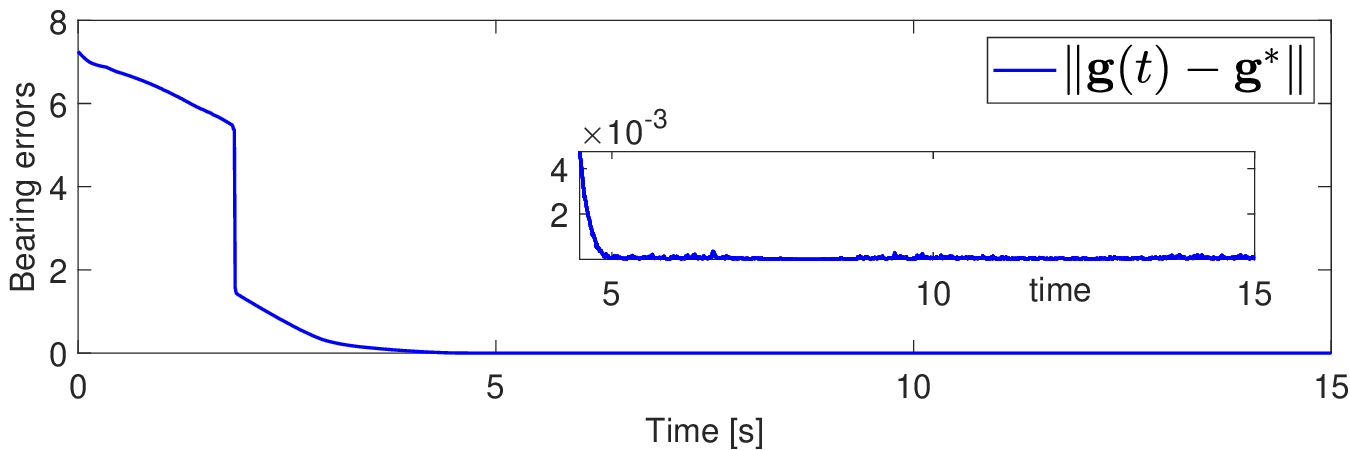}}
\caption{Simulation of six-agent system under the adaptive sliding-mode bearing-based control law \eqref{eq:bearing_based_control_law}.}
\label{fig:sim1}
\end{figure}
\begin{figure}[th]
\centering
\subfloat[]{\includegraphics[width=.45\textwidth]{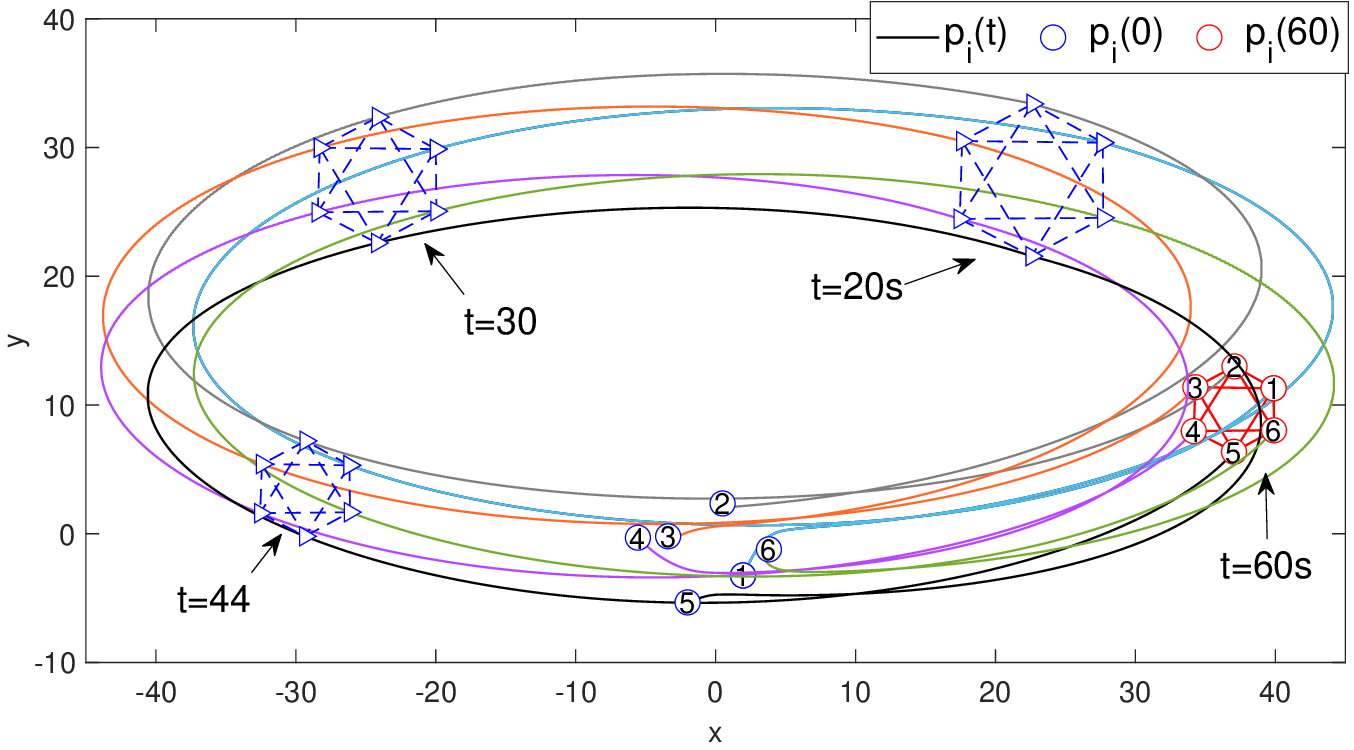}} \\
\subfloat[]{\includegraphics[width=.46\textwidth]{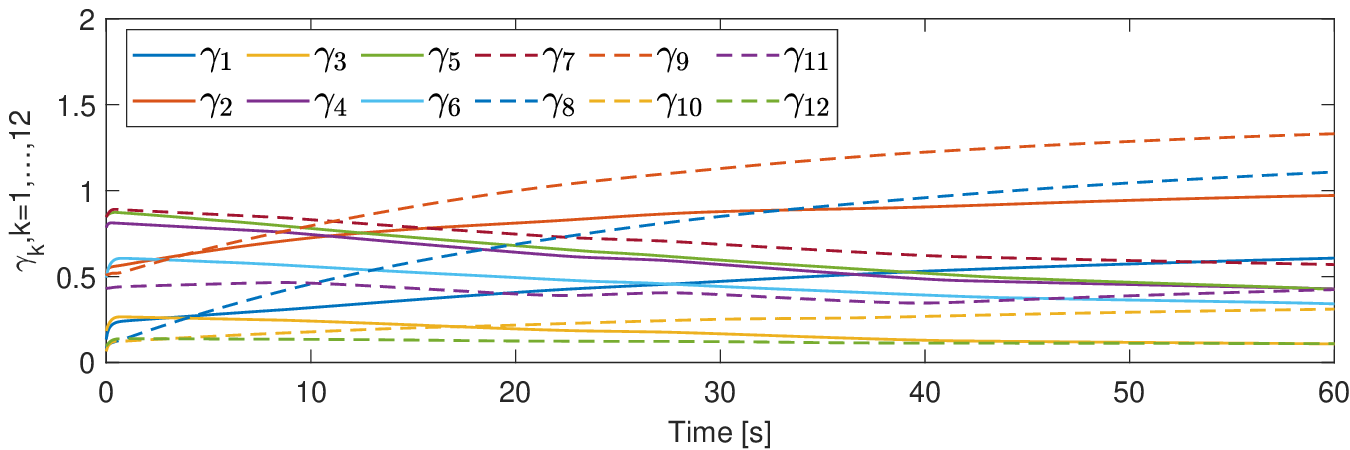}} \\
\subfloat[]{\includegraphics[width=.45\textwidth]{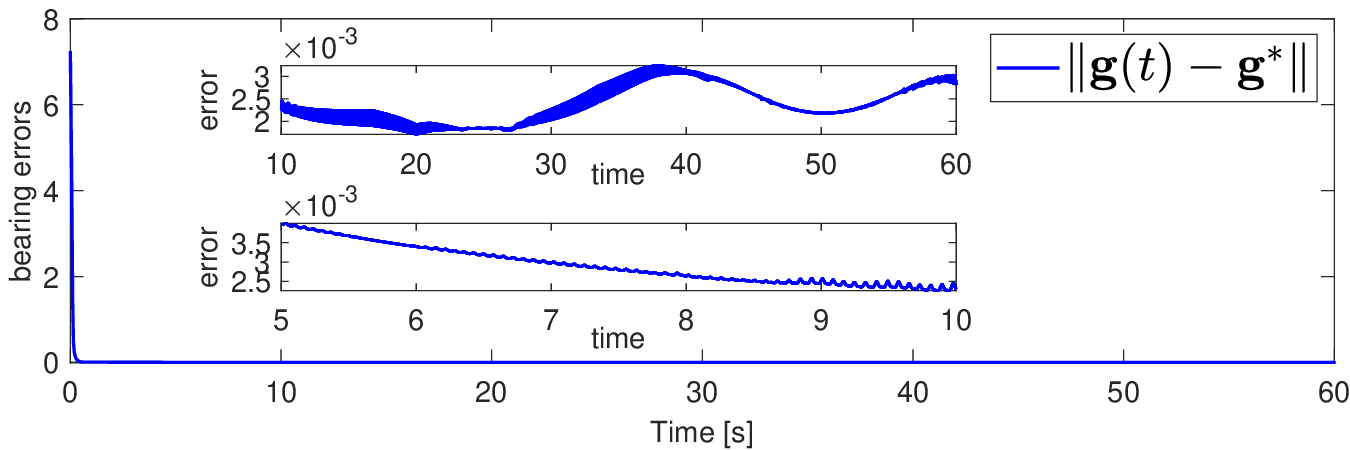}}
\caption{Simulation of six-agent system under the adaptive bearing-based control law \eqref{eq:Practical_Bearing_B}.}
\label{fig:sim1a}
\end{figure}
Simulation results for $t\in [0,50]$ (second) are given in Fig.~\ref{fig:sim1}. The matched disturbances are compensated along with the attainment of the set of desired bearing vectors. Bearing errors $\|\m{g}-\m{g}^*\|$ converges to 0 and $\gamma_k$ converge to some constant values for all $k=1,\ldots,12$ after about 5 seconds. The formation is driven by two disturbance terms $\m{d}_2$ (translation velocity $\m{f}_2$) and $\m{d}_3$ (formation rescales at rate $f_3(t)$) in the set of desired formations.

Second, we consider the system under the smooth adaptive bearing-based control law \eqref{eq:Practical_Bearing_B}, with $k_{\gamma}=0.2$, $\varepsilon=0.05$ and $\alpha=0.05$. The mismatched disturbance terms are changed correspondingly to $\m{f}_2=[5\cos(0.04\pi t),2\sin(0.04\pi t)]^\top$ and $f_3 = 0.25\left(3e^{-0.1t} +\sin\left(\frac{\pi t}{10}\right) \right)$.

Simulation results for $t\in [0,60]$ (seconds) are depicted in Fig.~\ref{fig:sim1a}. After 5 seconds, the bearing errors $\|\m{g}(t)-\m{g}^*\|$ are kept within the interval $[0,~5\cdot 10^{-3}]$. The adaptive gains do not converge to fixed values, and also vary within a positive, and bounded interval. The formation moves with velocity $\m{f}_2$ and changes the scale due to the mismatched disturbances $\m{d}_2$ and $\m{d}_3$, respectively.

\subsection{Simulation 2: Bearing-only control}
In this subsection, we first consider the six-agent system under the adaptive sliding-mode bearing-only control law \eqref{eq:Bearing_OnlyC}, with $k_{\gamma}=0.25$. The functions $\m{f}_1,\m{f}_2,f_3$ are selected the same as in Simulation 1, but the disturbances' components are $\m{d}_1 = \m{R}_{\rm b}^\top\m{f}_1$ and $\m{d}_3 = f_3 \hat{\m{r}}$. From the simulation results displayed in Fig.~\ref{fig:sim2}, the adaptive gains converge to constant values, the desired bearing constraints are almost achieved after 10 seconds, and the bearing error is kept smaller than $5\cdot 10^{-4}$ for $t\geq 10$. The motions of the formation consists of a translation with velocity $\m{f}_2$ and a scaling motion $f_3\hat{\m{r}}$. .

\begin{figure}[t]
\centering
\subfloat[]{\includegraphics[width=.45\textwidth]{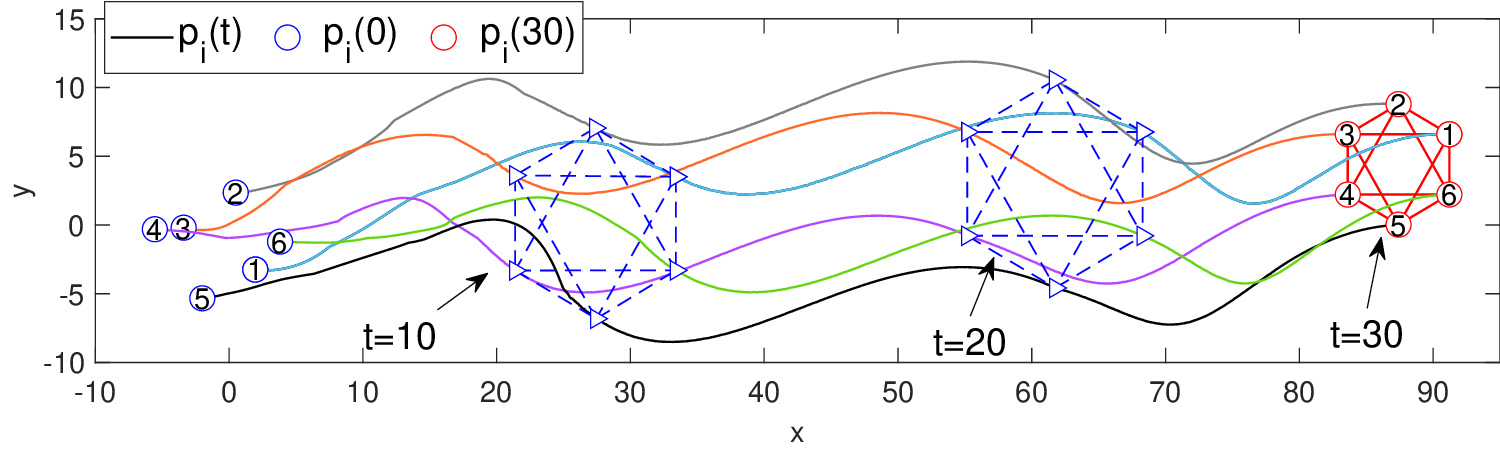}} \\
\subfloat[]{\includegraphics[width=.465\textwidth]{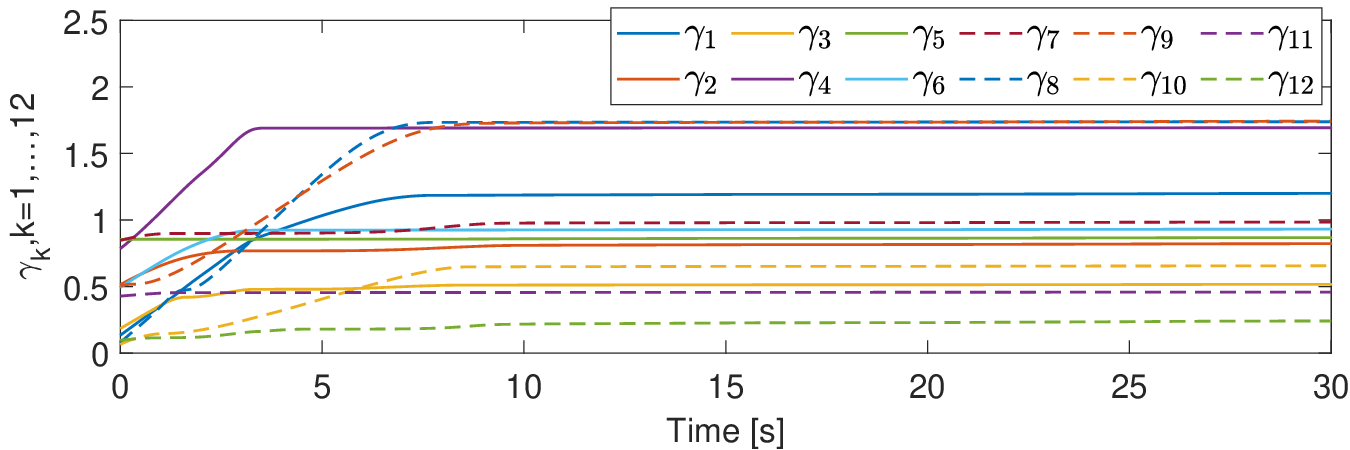}} \\
\subfloat[]{\includegraphics[width=.45\textwidth]{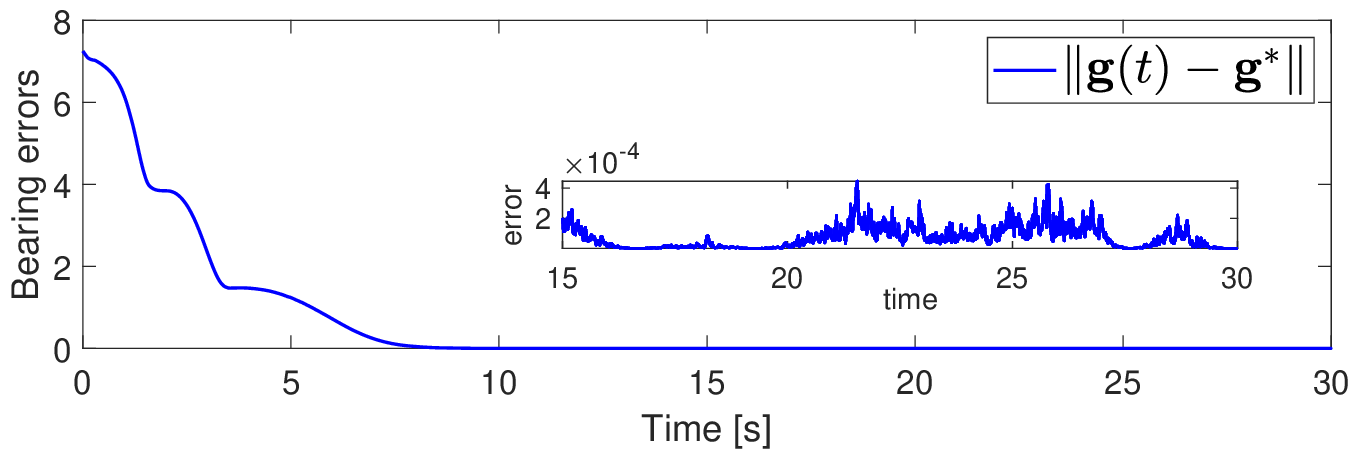}}
\caption{Simulation of six-agent system under the adaptive sliding-mode bearing-only control law \eqref{eq:Bearing_OnlyC}.}
    \label{fig:sim2}
\end{figure}

Second, we consider the six-agent system under the smooth adaptive bearing-only control law \eqref{eq:Practical_Bearing_OnlyC}. The disturbance components are changed to $\m{f}_2 = [2.5\cos(0.015\pi t),\sin(0.015\pi t)]^\top$ and $f_3=0.75e^{-0.1t} + \sin(0.015\pi t)$. 
\begin{figure}[th!]
    \centering
    \subfloat[]{\includegraphics[width=.45\textwidth]{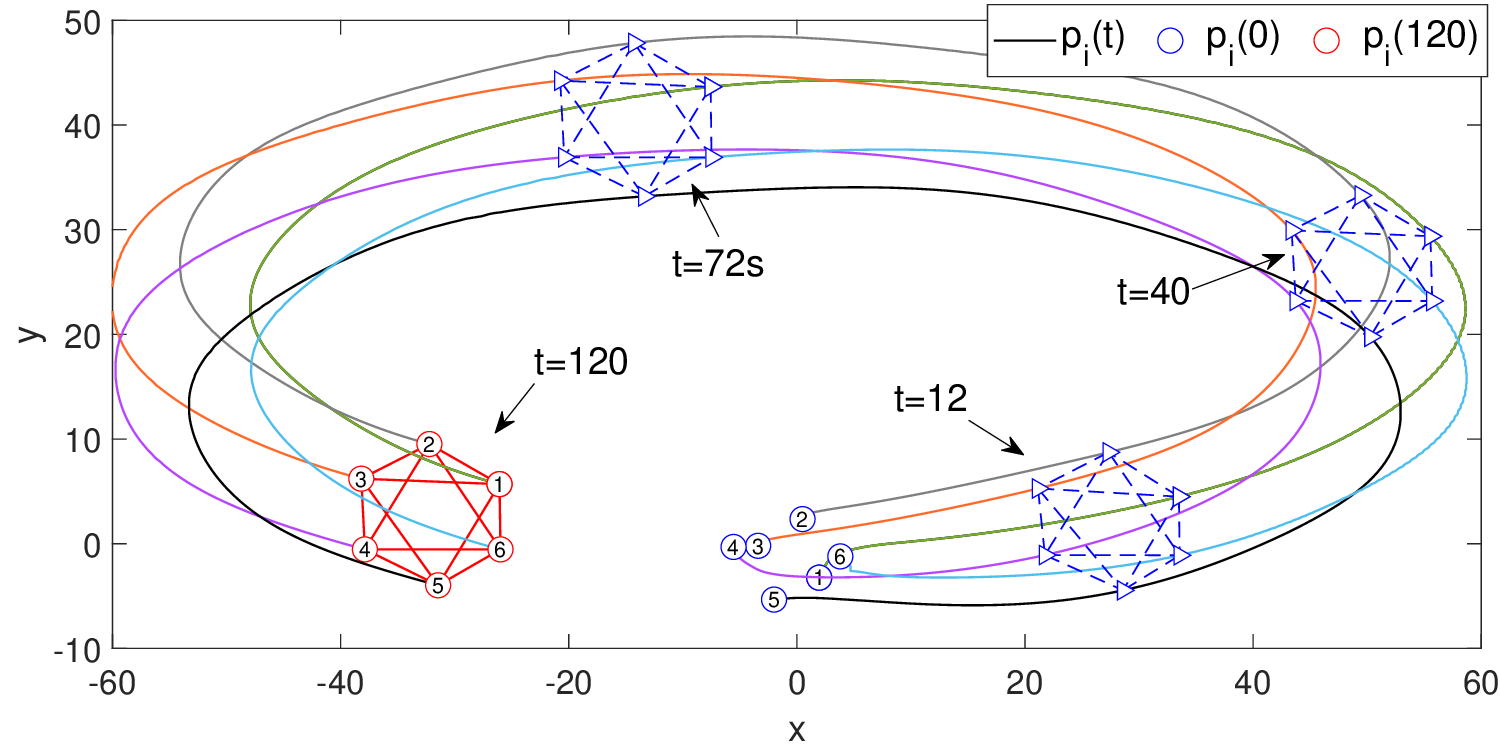}} \\
    \subfloat[]{\includegraphics[width=.46\textwidth]{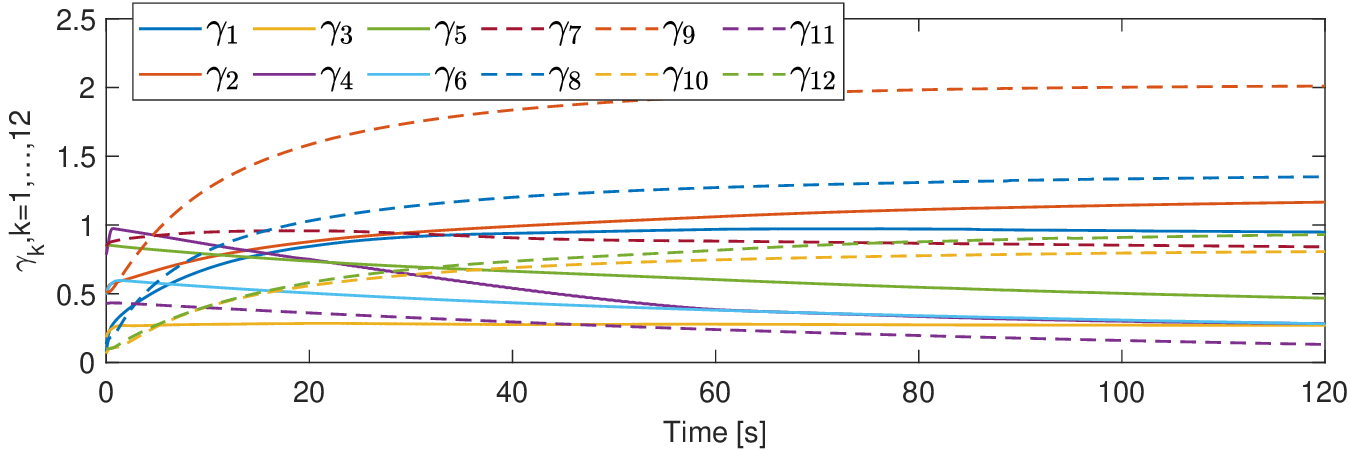}} \\
    \subfloat[]{\includegraphics[width=.45\textwidth]{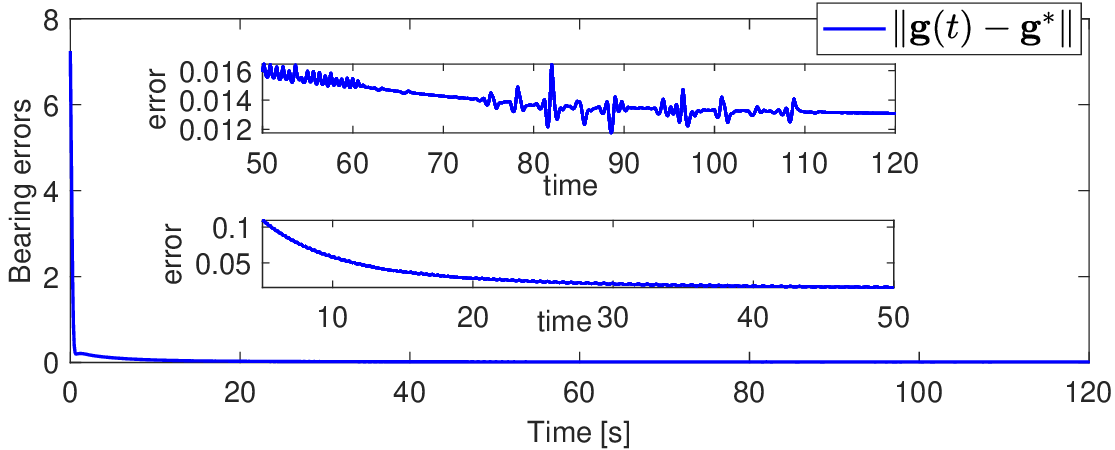}}
    \caption{Simulation of six-agent system under the smooth bearing-only control law \eqref{eq:Practical_Bearing_OnlyC}.}
    \label{fig:sim2a}
\end{figure}
With the control parameters selected as $k_p=4$, $k_{\gamma}=0.5$, $\varepsilon = 0.03$ and $\alpha = 0.02$, the simulation results of the system for $0\leq t\leq 120$ are given in Fig.~\ref{fig:sim2a}. Observe that the bearing error is maintained smaller than 0.05 for $t\geq 20$ (second). The adaptive gains do not converge, and vary within some positive, small intervals. Thus, we conclude that simulation results are consistent with the analysis.

\section{Conclusions}
\label{sec:6}
This paper investigates bearing-constrained formation control under unknown bounded disturbances using both displacement-based and bearing-only measurements. The disturbances were decomposed into matched and mismatched components, which introduce deterministic time-varying perturbations to the target formation. Sliding-mode controllers with state-dependent and time-dependent adaptation laws were proposed to uniformly stabilize the set of desired formations and completely reject the matched disturbance. Furthermore, simple smooth adaptive control laws were proposed to ensure uniform ultimate boundedness of the target formation and eliminate chattering phenomenon. While this paper confines the focus on bearing-constrained formation, the adaptive control design strategies in this paper are applicable to other categories of  formation control schemes. 


\bibliographystyle{IEEEtran}
\bibliography{IEEEabrv,Ref}

\appendix

\section{Proofs}
\subsection{Proof of Theorem \ref{thm:3.2}}
\label{app:thm3_2}
Using the Lyapunov function $V = \frac{1}{2} \|\bm{\delta}\|^2 + \frac{1}{2k_{\gamma}} \|\bm{\gamma} - \beta_4 \m{1}_n\|^2$, for some $\beta_4 > \max_{k=1,\ldots,m} \sup_{t\geq 0}\|\m{f}_k\|$, which is continuously differentiable, positive definite and radially unbounded with regard to $[\bm{\delta}^\top,(\bm{\gamma} - \gamma^* \m{1}_n)^\top]^\top$, we have
\begin{align}
    \dot{V} &= -k_p\bm{\delta}^\top\m{L}_{\rm b}^*\bm{\delta}  + \frac{1}{k_{\gamma}}\sum_{k=1}^n(\gamma_k-\beta_4)\dot{\gamma}_k \nonumber\\
    &\qquad\qquad- \bm{\delta}^\top\bar{\m{H}}^\top\bar{\m{P}}_{\m{g}^*}\left( \bar{\m{P}}_{\m{g}^*}^\epsilon \bar{\m{H}}\bm{\delta} +\m{f}_1\right) \nonumber\\
    &=-k_p \bm{\delta}^\top \m{L}_{\rm b}^* \bm{\delta} - \sum_{k=1}^m \left(\gamma_k \frac{\|\m{q}_{k}\|^2}{\|\m{q}_{k}\|+\varepsilon} - \m{q}_k^\top \m{f}_k \right)\nonumber\\
    &\qquad\qquad   + \sum_{k=1}^m (\gamma_k - \beta_4) \left(\frac{\|\m{q}_{k}\|^2}{\|\m{q}_{k}\|+\varepsilon} -\alpha\gamma_k \right)  \nonumber\\
    &\leq -k_p \lambda_{d+2}(\m{L}_{\rm b}^*) \|\bm{\delta}\|^2 -\sum_{k=1}^m \left(\beta_4\frac{\|\m{q}_{k}\|^2}{\|\m{q}_{k}\|+\varepsilon} - \|\m{q}_k\|\|\m{f}_k\| \right)  \nonumber \\
    &\qquad\qquad 
    -\sum_{k=1}^m\left(\frac{\alpha}{2}(2\gamma_k^2-2\beta_4\gamma_k+\beta_4^2)-\frac{\alpha}{2}\beta_4^2\right)
    \nonumber\\
    &\leq -k_p \lambda_{d+2}(\m{L}_{\rm b}^*) \|\bm{\delta}\|^2-\sum_{k=1}^m \left(\frac{\alpha}{2}(\gamma_k-\beta_4)^2-\frac{\alpha}{2}\beta_4^2 \right)   \nonumber\\
    &\qquad\qquad -\sum_{k=1}^m \beta_4\|\m{q}_{k}\|\left( \frac{\|\m{q}_{k}\|}{\|\m{q}_{k}\|+\varepsilon} - 1 \right)  \nonumber\\
    &\leq -k_p \lambda_{d+2}(\m{L}_{\rm b}^*) \|\bm{\delta}\|^2 - \frac{\alpha}{2} \sum_{k=1}^m (\gamma_k-\beta_4)^2 + \frac{1}{2}m\alpha\beta_4^2 \label{eq:dotV_BB_pract}
\end{align}
Let $\zeta\triangleq \min\{2k_p\lambda_{d+1}(\m{L}_{\rm b}^*),k_{\gamma}\alpha\}$ and $\Delta\triangleq \frac{1}{2}m\alpha\beta_4^2$, we rewrite \eqref{eq:dotV_BB_pract} as $\dot{V}\leq -\zeta V + \Delta$. By comparison lemma \cite{Khalil2002nonlinear}, $V(t) \leq \frac{\Delta}{\zeta} + {\rm exp}(-\zeta t)\left(V(0)-\frac{\Delta}{\zeta}\right)$. This implies global uniform ultimate boundedness of the equilibrium $[\bm{\delta}^\top,(\bm{\gamma} - \gamma^* \m{1}_n)^\top]^\top=\m{0}_{dn+m}$, with the ultimate bound $\sqrt{\zeta^{-1}\Delta+\epsilon_1}$, for some small positive $\epsilon_1$ \cite{Khalil2002nonlinear}.

\subsection{Proof of Theorem \ref{thm:4.2}}
\label{app:thm4_2}
Consider the Lyapunov function $V=\frac{1}{2}\|\bm{\delta}\|^2 + \frac{1}{2k_{\gamma}}\sum_{k=1}^m\|\m{z}^d_k\|(\gamma_{k}-\beta_7)^2$, where $\beta_7> \max_{k}\sup_{t\geq 0}\|\m{f}_k\|$. By denoting $\bm{\eta}_k \triangleq \m{P}_{\m{g}_k}\m{g}_k^*$, with $\|\bm{\eta}_k\|\leq 1$, we have,
\begin{align}
\dot{V} &=\bm{\delta}^\top\tilde{\m{R}}_{\rm b}^\top\left(k_p\m{g}^*+\bar{\bm{\Gamma}}{\m{P}}_{\m{g}}^{\varepsilon}\m{g}^* + \m{f} \right) + \sum_{k=1}^m \frac{\|\m{z}_k^d\|}{k_{\gamma}}(\gamma_k-\beta_7)\dot{\gamma}_k \nonumber\\
&=-k_p\sum_{k=1}^m\|\m{z}_k^d\|\bm{\eta}_k^\top\bm{\eta}_k - \sum_{k=1}^m\|\m{z}_k^d\| \left(\gamma_k\frac{\bm{\eta}_k^\top\bm{\eta}_k}{\|\bm{\eta}_k\|+\varepsilon} + \bm{\eta}_k^\top\m{f}_k\right) \nonumber\\
&\qquad\qquad + \sum_{k=1}^m \|\m{z}_k^d\|(\gamma_k - \beta_7) \left(\frac{\|\bm{\eta}_k\|^2}{\|\bm{\eta}_k\|+\varepsilon} -\alpha\gamma_k\right) \nonumber\\
&\leq -\sum_{k=1}^m\|\m{z}_k^d\|\left(k_p\|\bm{\eta}_k\|^2+\frac{\alpha}{2}(2\gamma_k^2-2\beta_7\gamma_k+\beta_7^2)-\frac{\alpha}{2}\beta_7^2\right) \nonumber\\
&\qquad\qquad -\sum_{k=1}^m \|\m{z}_k^d\|\left(\beta_7\frac{\|\bm{\eta}_k\|^2}{\|\bm{\eta}_k\|+\varepsilon}-\|\bm{\eta_k}\|\|\m{f}_k\|\right) \nonumber\\
&\leq-\sum_{k=1}^m \|\m{z}_k^d\|\left(k_p\|\bm{\eta}_k\|^2 + \frac{\alpha}{2}(\gamma_k-\beta_7)^2-\frac{\alpha}{2}\beta_7^2 \right) \nonumber\\
&\qquad\qquad - \sum_{k=1}^m \beta_7 \|\m{z}_k^d\| \|\bm{\eta}_k\| \left(\frac{\|\bm{\eta}_k\|}{\|\bm{\eta}_k\|+\varepsilon} -1 \right) \label{eq:dotV_practice1}
\end{align}
In the last term in the right-hand side of Eq.~\eqref{eq:dotV_practice1}, by using  $\frac{\|\bm{\eta}_k\|}{\|\bm{\eta}_k\| + \varepsilon}< 1$, we have
\begin{align}
\dot{V} & -k_p\sum_{k=1}^m\|\m{z}_k^d\|\|\bm{\eta}_k\|^2 -\frac{\alpha}{2}\sum_{k=1}^m\|\m{z}_k^d\|(\gamma_k-\beta_7)^2 \nonumber\\
&\qquad\qquad + \underbrace{\frac{\alpha}{2}\beta_7^2\sum_{k=1}^m\|\m{z}_k^d\| + \varepsilon\beta_7\sum_{k=1}^m \|\m{z}_k^d\|}_{\triangleq \Delta}. \label{eq:dotV_practice}
\end{align}
With a similar argument as in the proof of Theorem~\ref{thm:2}, we can derive a bound on the first term in the right-hand side of \eqref{eq:dotV_practice} and eventually obtain
\begin{align}
\dot{V} \leq -\zeta V + \Delta,
\end{align}
where $\zeta\triangleq \min\{k_p\lambda_{d+2}(\m{L}_{\rm b}^*)\frac{\min_{k=1,\ldots,m}\|\m{z}_k^d\|}{2(n-1)s^2},k\alpha\}$. It follows from comparison lemma \cite{Khalil2002nonlinear} that for an arbitrarily small $\epsilon>0$, there exists a finite time $T_0\geq 0$ such that $V(t)\leq \frac{\Delta}{\zeta}+\epsilon,\; \forall t\geq T_0$. This implies that $\bm{\delta}$ and $\bm{\gamma}$ are both uniformly ultimately bounded, with the ultimate bound $\sqrt{2(\zeta^{-1}\Delta+\epsilon)}$ independent of the initial condition.

\end{document}